\documentclass[twocolumn,usegraphicx,usenatbib]{mn2e}

\usepackage{subfig}
\usepackage{threeparttable}
\usepackage{amsmath}

\newcommand{\mnras}{MNRAS}
\newcommand{\aap}{A\&A}
\newcommand{\prd}{PRD}
\newcommand{\apj}{ApJ}
\newcommand{\aj}{AJ}
\newcommand{\araa}{ARA\&A}
\newcommand{\apjs}{ApJS}
\newcommand{\apjl}{ApJ}
\newcommand{\revmodphys}{Rev. Mod. Phys.}
\newcommand{\prl}{PRL}
\newcommand{\compphyscomm}{Comp. Phys. Comm.}
\newcommand{\nat}{Nat}
\newcommand{\sci}{Sci}
\newcommand{\baas}{BAAS}
\newcommand{\modsec}[1]{#1}

\newtheorem{defn}{Definition}
\newtheorem{thm}{Theorem}
\newtheorem{rem}{Remark}
\newtheorem{corrol}{Corollary}

\usepackage{times}

\input epsf

\begin{document}
\title[Parameters from the XCS]{The {\it XMM} Cluster Survey: Forecasting cosmological and cluster scaling-relation parameter constraints}
\author[M.~Sahl\'en et al.]{Martin
Sahl\'en,$^1$\thanks{m.sahlen@sussex.ac.uk} Pedro T.~P.~Viana,$^{2,3}$
Andrew R.~Liddle,$^{1}$ A.~Kathy Romer,$^{1}$ \cr Michael
Davidson,$^{4}$ Mark
Hosmer,$^{1}$ Ed Lloyd-Davies,$^{1}$ Kivanc Sabirli,$^{5}$ \cr
Chris A.~Collins,$^6$ Peter E.~Freeman,$^5$ Matt Hilton,$^{7,8}$ Ben
Hoyle,$^9$
Scott T.~Kay,$^{10}$ \cr
Robert G.~Mann,$^{4}$ Nicola Mehrtens,$^1$
Christopher J.~Miller,$^{11}$ Robert C.~Nichol,$^9$  \cr
S.~Adam Stanford$^{12,13}$ and
Michael J.~West$^{14,15}$ (The XCS Collaboration)\thanks{http://xcs-home.org} 
\\
\vspace*{-6pt} {\small \em $^1$Astronomy Centre, University of Sussex,
Falmer, Brighton BN1 9QH, UK}\\
\vspace*{-6pt} {\small \em $^2$Departamento de Matem\'{a}tica Aplicada da
Faculdade de Ci\^{e}ncias da Universidade do Porto, Rua do Campo
Alegre, 687, 4169-007 Porto, Portugal}\\
\vspace*{-6pt} {\small \em $^3$Centro de Astrof\'{\i}sica da
Universidade do Porto, Rua das Estrelas, 4150-762 Porto, Portugal}\\
\vspace*{-6pt} {\small \em $^4$SUPA, Institute for Astronomy, University of Edinburgh, Blackford Hill, Edinburgh, EH9 9HJ, UK}\\
\vspace*{-6pt} {\small \em $^{5}$Department of Physics, Carnegie Mellon
University, 5000 Forbes Avenue, Pittsburgh, PA-15217, USA}\\
\vspace*{-6pt} {\small \em $^{6}$Astrophysics Research Institute,
Liverpool John Moores University, Twelve Quays House, Egerton Wharf,
Birkenhead CH41 1LD, UK}\\
\vspace*{-6pt} {\small \em $^{7}$Astrophysics \& Cosmology Research Unit, School of Mathematical Sciences, University of KwaZulu--Natal, Private Bag X54001, Durban 4000, S. Africa} \\
\vspace*{-6pt} {\small \em $^{8}$South African Astronomical Observatory, PO Box 9, Observatory, 7935, Cape Town, S. Africa} \\
\vspace*{-6pt} {\small \em $^9$ICG, University of Portsmouth,
Portsmouth PO1 2EG, UK}\\
\vspace*{-6pt} {\small \em $^{10}$Jodrell Bank Centre for Astrophysics,
School of Physics and Astronomy, The University of Manchester, Manchester M13 9PL, UK}\\
\vspace*{-6pt} {\small \em $^{11}$Cerro-Tololo Inter-American Observatory,
National Optical Astronomy Observatory, 950 North Cherry Avenue, Tucson,
AZ 85719, USA}\\
\vspace*{-6pt} {\small \em $^{12}$Department of Physics, University of
California at Davis, 1 Shields Avenue, Davis, CA 95616-8677, USA}\\
\vspace*{-6pt} {\small \em $^{13}$Institute of Geophysics \&
Planetary Physics,
Lawrence Livermore National Laboratory, L-413, P.O. Box 808, 7000 E.
Avenue, Livermore, CA 94551, USA}\\
\vspace*{-6pt} {\small \em $^{14}$Gemini Observatory, Casilla 603, La
Serena, Chile}\\
\vspace*{-6pt} {\small \em $^{15}$European Southern Observatory,
Alonso de C\'ordova 3107, Vitacura, Casilla 19001, Santiago 19, Chile}
}
\maketitle
\begin{abstract}
We forecast the constraints on the values of $\sigma_8$, $\Omega_{\rm
m}$, and cluster scaling relation parameters which we expect to obtain from
the {\it XMM} Cluster Survey (XCS). We assume a flat $\Lambda$CDM Universe and
perform a Monte Carlo Markov Chain analysis of the evolution of the number density of
galaxy clusters that takes into account a detailed simulated selection
function. Comparing our current observed number of clusters shows good agreement with predictions.
We determine the expected degradation of the constraints
as a result of self-calibrating the luminosity--temperature relation
(with scatter), including temperature measurement errors, and relying on
photometric methods for the estimation of galaxy cluster redshifts.
We examine the effects of systematic errors in scaling relation and
measurement error assumptions. Using only $(T,z)$ self-calibration, we expect to measure
$\Omega_{\rm m}$ to $\pm 0.03$ (and $\Omega_\Lambda$ to the
same accuracy assuming flatness), and $\sigma_8$ to $\pm 0.05$, also constraining the normalization and slope of the luminosity--temperature relation to $\pm 6$ and $\pm 13$ per cent (at $1\sigma$) respectively in the process. Self-calibration fails to jointly constrain the scatter and redshift evolution of the luminosity--temperature relation significantly. Additional archival and/or follow-up data will improve on this. We do not expect measurement errors or imperfect knowledge of their distribution to degrade constraints significantly. Scaling-relation systematics can easily lead to cosmological constraints $2\sigma$ or more away from the fiducial model. Our treatment is the first exact treatment to this level of detail, and introduces a new `smoothed ML' estimate of expected constraints.
\end{abstract}

\section{Introduction}
\label{Introduction}
The abundance of galaxy clusters as a function of mass and redshift
can give a powerful constraint on cosmological models.  Specifically,
data on the evolution of the number density of galaxy clusters with
redshift has been used to obtain direct estimates for both
$\sigma_{8}$, the dispersion of the mass field smoothed on a scale of
$8\,h^{-1}\;{\rm Mpc}$, and on $\Omega_{\rm m}$, the
present mean mass density of the Universe (\citealt{FWED,OB,VL96,OB97,H97}; \citealt*{BFC}; \citealt{Ekeetal,Retal,DV,
VL99,BSBL,H00,Betal}; \citealt*{RVP}; \citealt{Henry04,Getal06,Rozomaxbcg}).
Furthermore, such data could be used to constrain the present energy density of a dark energy
component, $\Omega_{w}$, and its equation of state
(\citealt{WS}; \citealt*{HMH}; \citealt{HT}; \citealt*{LSW}; \citealt*{WBK}; \citealt{BW,Hu,MJ03,MJ04,Wang04,LH,Mantz:2007qh}), or more
simply the present vacuum energy density associated with a
cosmological constant, $\Omega_{\Lambda}\equiv\Lambda/3H^{2}_{0}$
\citep*{HHM}.
Others have
suggested using galaxy clusters to constrain particle physics beyond
the Standard Model (e.g. \citealt{Wang05}; \citealt*{EGW}), or modified-gravity models
where it has been shown that e.g. the Dvali--Gabadadze--Porrati (DGP) modified-gravity model should be testable in coming surveys \citep{Tang,Schaefer}.  An alternative method to abundance evolution using X-ray galaxy clusters to constrain cosmology, is based on the gas mass fraction (e.g. \citealt*{Allen02}; \modsec{\citealt{Viketal03};} \citealt*{ER03}; \citealt*{RAW}; \citealt{Vik05b,Allen:2007ue, Rapetti:2007mw}).

Galaxy cluster measurements
are complementary to other cosmological constraints derived from the
Cosmic Microwave Background (CMB) and distant Type Ia Supernovae
observations, and thus help break degeneracies amongst the various
cosmological parameters
\citep{BOPS,HMH,HT,LSW,BW,MBBS,Wang04}.

Several surveys have
been proposed with the explicit aim of significantly increasing the
number of known distant clusters of galaxies. These proposals rely on
a variety of detection methods across a wide range of wavelengths: the
Sunyaev--Zel'dovich (SZ) effect in the millimeter (see \citealt*{CHR} for a
review, and \citealt{Juin} for a list of proposed surveys); galaxy
overdensities in the visible/infrared (e.g. \citealt{GY05,Hsieh,Rozooptical}); bremsstrahlung emission by the intracluster medium (ICM) in the X-rays \citep[e.g.][]{DUET,HaimanDETF,XMMLSSfuture}.  Galaxy
cluster identification using weak lensing techniques is another
possibility \citep[e.g.][]{Wittman}, but is still in its
infancy. Many of these proposals, in particular those
regarding the detection of distant clusters through their
X-ray emission, imply the building of new observing facilities such as
{\it eROSITA} \citep{erosita}, that will likely take many years to yield
results. The cluster X-ray temperature
is one of the best proxy observables in lieu of mass; it is a better estimator
of the cluster mass than the cluster X-ray luminosity
\citep[e.g.][]{Bal,Zhang} (but more difficult to determine), and galaxy clusters are also most unambiguously identified in X-ray images. This makes X-ray-based
galaxy cluster surveys those with the most accurately determined
selection function.  For all these reasons, we have undertaken to construct a
 galaxy cluster catalogue, called {\em XCS: XMM Cluster
Survey}, based on the serendipitous identification of galaxy clusters
in public {\it XMM--Newton} ({\it XMM}) data \citep{RVLM}.

The aim of this paper is to forecast the expected galaxy cluster samples from the XCS and, based on those, its ability to constrain cosmology and cluster scaling relations using only self-calibration. Specifically, we consider the expected constraints on $\Omega_{\rm m}$, $\sigma_8$ and the luminosity--temperature relation for a flat Universe. Our results represent the statistical power expected to be present in the full {\it XMM} archive.  This work builds upon previous efforts in several ways, and to a large extent constitutes the first coherent treatment of effects and methods previously only considered separately. Specifically, we combine all the following characteristics:
\begin{enumerate}
\item we use a Monte Carlo Markov Chain (MCMC) approach and can thus characterize all
  degeneracies exactly (in contrast to Fisher matrix analyses),
\item we include scatter in
  scaling relations in the parameter estimation (enabled by MCMC),
\item we include a detailed, simulated selection function (essentially that of the {\it XMM} archive), not a simple hard flux/photon-count/mass limit,
\item we include realistic photometric redshift errors, including degradation and catastrophic errors,
\item we include temperature measurement errors, partly based on detailed
  simulations of {\it XMM} observations, and propagate the redshift errors to the temperature, and,
\item we investigate quantitatively the effect on
  cosmological constraints from systematic errors in cluster scaling relation
  and measurement error characterization.
\end{enumerate}

Our work builds on the galaxy cluster survey exploitation methods developed and studied primarily in \citet{HMH,HHM,LSW,HK,Hu,BW,MJ03,MJ04,LHI,Wang04,LH}. Forecasted cosmological constraints from {\it XMM} data have also been considered for the {\it XMM}--LSS survey in \citet{RVP}, but they did not take into account scaling-relation scatter or measurement errors, and used the Press--Schechter mass function.
The most relevant precursors to this paper are \citet{HMH} and
\citet{MJ04}, who consider cosmological constraints expected from the {\it Dark Universe Exploration Telescope} ({\it DUET})
\citep{DUET} -- a $10000$ deg$^2$ X-ray survey with flux limit $\sim 5
\times 10^{-14}$ erg s$^{-1}$ cm$^{-2}$ in the $0.5$–-$2$ keV
band.
We extend the methodology of both papers through each of the six points above, either by more detailed modeling or by obtaining more robust results through the use of MCMC.
Other relevant works are \citet{Hut04,Hut06} and \citet{LimaHuphotoz}, who discuss photometric redshifts. We particularly complement these analyses through our detailed treatment/inclusion of measurement errors and selection effects.  The recent work by \citet{Rapetti:2007mw} takes an approach similar to ours in that they employ MCMC, include scaling-relation scatter and consider measurement errors,  but focuses on combining future X-ray gas mass fraction measurements with SZ cluster and CMB power spectrum data.

The structure of this paper is as follows. We begin by reviewing the progress to date of the XCS and present the survey selection function (Sect.~\ref{XCS}).  Next, we present the models and methodology we use to
derive constraints on cosmological parameters from the simulated XCS sample (Sects.~\ref{Connect} \& \ref{Methodology}). We then go on to the expected cluster distributions and, our estimates
for the constraints on $\sigma_{8}$, $\Omega_{\rm m}$, and cluster scaling
relation parameters that we expect to obtain from the XCS using
self-calibration, including the effect of temperature measurement
errors and relying on photometric methods to obtain XCS galaxy cluster
redshifts (Sect.~\ref{Results}).  We discuss and summarize our findings in Sect.~\ref{Conclusions}.  Additional material setting out modeling details is provided in the Appendix.

\section{The XMM Cluster Survey}
\label{XCS}

\subsection{Background and current status}
{\it XMM--Newton} is the most sensitive X-ray spectral imaging telescope
deployed to date. It is typically used in pointing mode, whereby it observes a single central target for a long period of time (the typical exposure time being $\sim 20$ kilo-seconds).
The field of view of the {\it XMM} cameras is roughly half a degree across, so that a considerable area around the central target is observed `for free' during these long pointings. Already many thousands of these pointings are available in the public {\it XMM} archive. The XCS is exploiting this archive by carrying out a systematic search for serendipitous detections of clusters of galaxies in the outskirts of {\it XMM} pointings \citep{RVLM}. Once a cluster candidate has been selected from the archival imaging data, it is then followed up using optical imaging and/or optical spectroscopy, to confirm the indentification of the X-ray source and to measure redshifts (see Sect.~\ref{Photozs}). For those XCS clusters that were detected with sufficient counts, an X-ray spectroscopy analysis is carried out, again using the archival data, in order to measure the temperature of the hot intracluster medium (ICM). These temperatures can then be used to study cluster scaling relations and/or to estimate the mass of the cluster (see Sects.~\ref{tempmass}~\&~\ref{lumtemp}).

\modsec{The XCS project is ongoing, but already more than $2000$ {\it XMM} pointings have been analysed, yielding a cluster candidate catalogue numbering almost $2000$ entries. So far, the XCS covers a combined area of $132$ deg$^{2}$ suitable for cluster searching and for which optical follow-up has been completed; i.e. this area excludes overlapping and repeat exposures, regions of low Galactic latitude, the Magellanic clouds, and pointings with very extended central targets. Around $75$--$100$ clusters with $>500$ photons and $T>2$ keV are present in this initial area.}
With many thousand more {\it XMM} pointings waiting to be analysed by the XCS, and a mission lifetime extending to $2013$, a conservative estimate for the final XCS area for cluster searching is $500$ deg$^{2}$. We use $500$ deg$^{2}$ herein for XCS cosmology forecasting (see Table~\ref{tab:surveys}), assume a redshift range of $0.1 \le z \le 1$, and temperatures of $2\,{\rm keV} \le T \le 8\,{\rm keV}$. We further limit our representative survey to clusters with photon counts $>500$ ($^{500}$XCS hereafter), so that we can be sure to estimate X-ray temperatures with reasonable accuracy (see Sect.~\ref{Terror}).  The lower redshift limit is associated with cluster extents becoming too large, and the cosmic volume also becoming small. The maximum redshift is chosen so that the luminosity--temperature relation can still be reliably modelled/estimated (see Sect.~\ref{ltscatter}). The temperature range is chosen such that we can expect i) a small contamination from galaxy groups (which typically have temperatures $T < 2\,{\rm keV}$), yet include as many of the numerous low-temperature clusters as possible, and ii) that clusters above the high-temperature limit are sufficiently rare that none can be expected.  The final cluster catalogue (without the cut-offs defined above for $^{500}$XCS) will contain several thousand clusters out to a redshift of $z \approx 2$. The highest-redshift cluster discovered by the XCS so far is XMMXCS~J2215.9-1738 at $z=1.457$  \citep{stanford2215,hilton2215}.

In addition to producing one of the largest samples of X-ray clusters ever compiled, the XCS will also be a valuable resource for cosmology studies (see Sect.~\ref{Methodology}). This is because the catalogue will be accompanied by a complete description of the selection function. In this work we make use of an initial XCS selection function that assumes simple models for the distribution of the ICM, and flat cosmologies (see below). Future cosmology analyses will take advantage of more sophisticated selection functions that are based on hydrodynamical simulations of clusters \citep{Kay07}.

\begin{table}
\centering
\begin{threeparttable}
\begin{tabular}{|l|c|}
  \hline
  Survey & $^{500}$XCS \\
  \hline
  Sky coverage & 500 deg$^2$ (serendipitous) \\
  Redshift coverage & 0.1 -- 1.0 \\
  X-ray temperature coverage & 2 -- 8 keV \\
  Min. photon count & 500\\
  X-ray flux limit &
   By selection function\tnote{a}
     \\
  \hline
\end{tabular}
\begin{tablenotes}
  \item [a] The flux limit is $\sim 3.5 \times 10^{-13}\;{\rm erg\,s^{-1}\,cm^{-2}}$ in the $[0.1,2.4]$ keV band, if defined as a probability of detection greater than or equal to $50$ per cent. See also Sect.~\ref{Distribnoerr} and Fig.~\ref{fig:xlimits}.
\end{tablenotes}
\caption{Survey specifications.}
\label{tab:surveys}
\end{threeparttable}
\end{table}

\begin{figure*}
\centering
\subfloat[Constant $L$--$T$
relation]{\includegraphics[width=8cm]{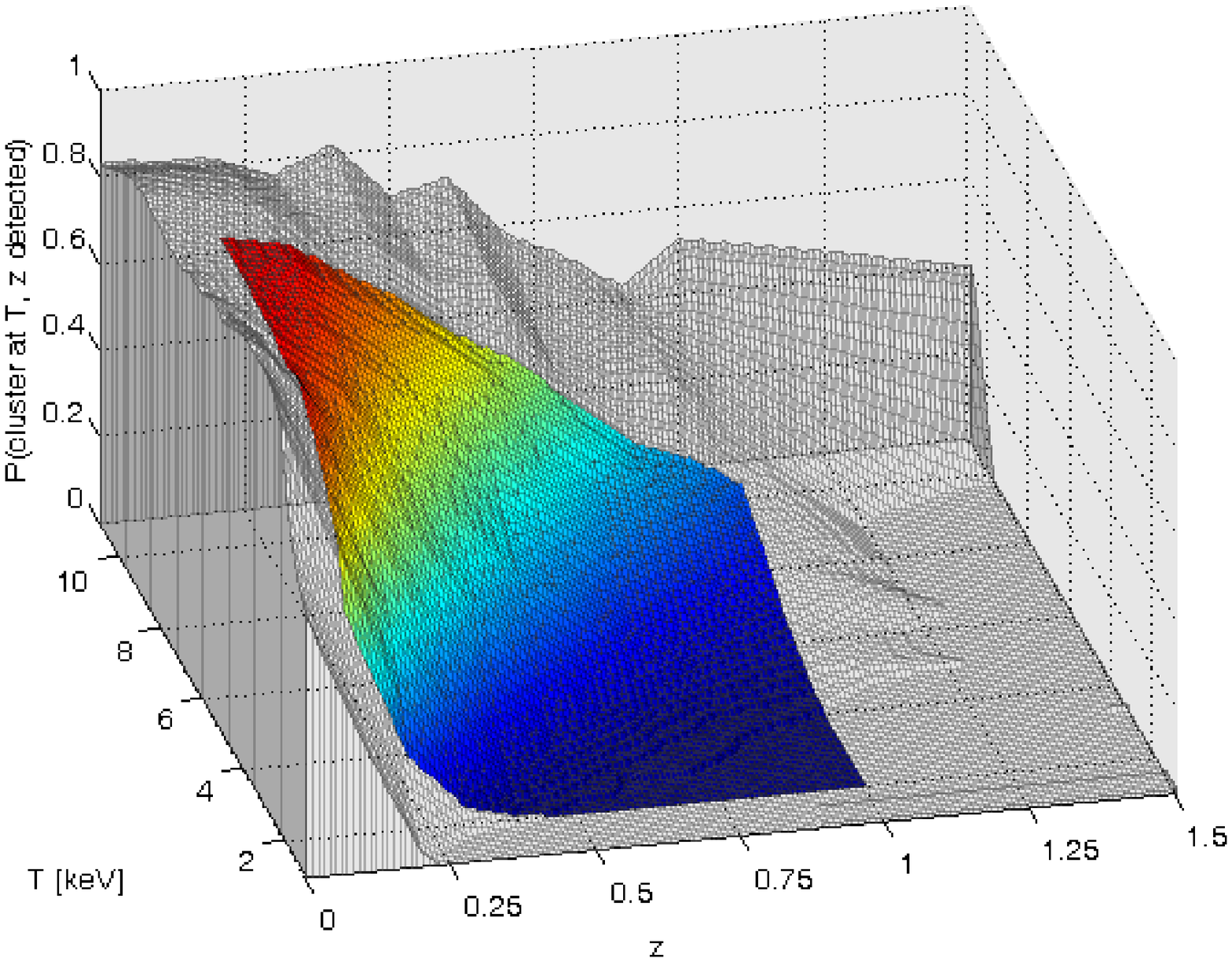}}
\subfloat[Self-similar $L$--$T$
relation]{\includegraphics[width=8cm]{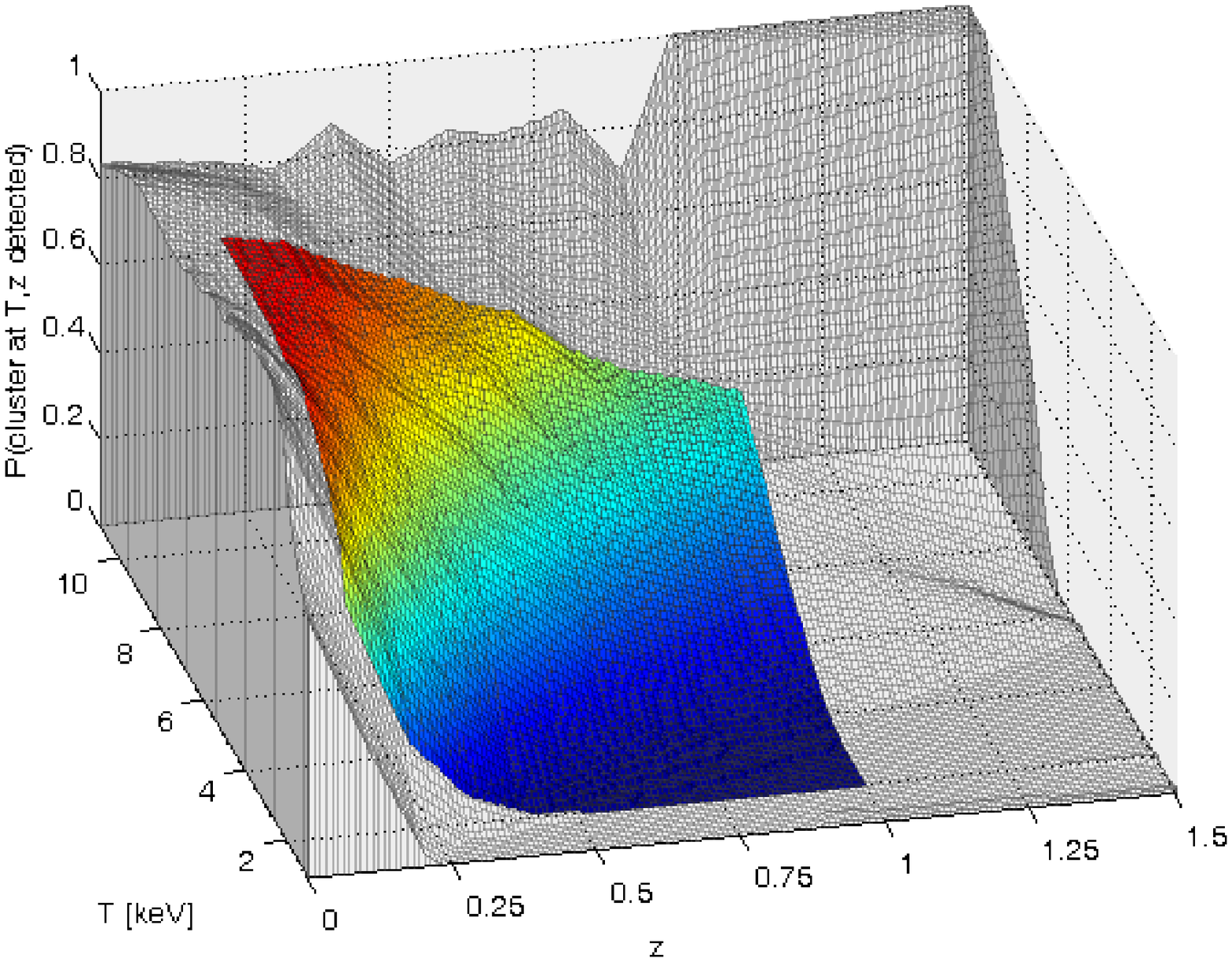}}
\caption{Selection function for our fiducial cosmology and different
$L$--$T$ evolution. Values in the shaded region are extrapolated from
those in the coloured region ($0.1~\le~z~\le~1.0,\,2\,{\rm keV}\le T
\le 8\,{\rm keV}$), for which the selection function has been
calculated explicitly.}
\label{fig:selfun}
\end{figure*}

\subsection{The XCS selection function}
\label{Selfun}

\subsubsection{Model}
\label{selfunmodel}
In order to properly model the selection function of a survey like the XCS, it is important to account for all of the observational variations present in real data. We can achieve this by placing a sample of fake surface-brightness profiles into real {\it XMM} Observation Data Files (ODFs). This ensures that our simulated images re-create real-life issues such as clusters lying on chip gaps and point-source contamination. The fake surface-brightness profiles are created as follows.  We use an isothermal $\beta$ model with $\beta = 2/3$, core radius $r_{\rm c} = 160\,{\rm kpc}$ (close to the mean values of $\beta = 0.64$, $r_{\rm c} =163\,{\rm kpc}$ obtained from a uniform {\it ROSAT} analysis of clusters from $0.1<z<1.0$; \citealp{OM04}), and plasma metallicity $Z = 0.3{\rm Z_{\sun}}$.
For a given cosmology we simulate $700$ sets of cluster parameters:
\begin{itemize}
\item $10$ redshifts (linearly spaced $0.1$--$1.0$)
\item $10$ luminosities (log. spaced $0.178$--$31.623 \times 10^{44}$ erg s$^{-1}$)
\item $7$ temperatures (linearly spaced $2$--$8$ keV)
\end{itemize}
For selection function determination, we drew on a list of $1764$ ODFs that have already been processed by the XCS and have been deemed to be suitable for cluster searching (see above). Before each selection function run, a smaller list of $100$ ODFs is selected at random from the full set of $1764$. These $100$ ODFs are then copied from the main XCS archive to local processing nodes for temporary storage, to speed up the analysis. Tests have shown that with $100$ ODFs it is still possible to reproduce the variance in exposure time, target type, point source density, etc., inherent to the XCS.  In the following we define a `selection function run' as the analysis over the $700$ sets of cluster parameters and $100$ ODFs -- a total of $70000$ combinations.

For each of the $700$ different combinations of cluster parameters, the process proceeds as follows. First, to account for the fact that the XCS searches the entire field of view for serendipitous cluster detections, the centre of the fake surface-brightness profile is randomly positioned into a blank {\it XMM}--style ODF, with a uniform probability across the field of view. The profile is then convolved with the appropriate PSF model. For this purpose we use the two-dimensional medium-accuracy model\footnote{http://xmm.vilspa.esa.es/external/xmm\_sw\_cal/calib/}.
At this stage, an ODF is chosen at random from the list of $100$ stored locally, into which the fake source will later be added. The profile is then assigned an absorbed count rate using a series of arrays calculated using X{\sc spec} \citep{arnaud96a}. The arrays tabulate conversions from unabsorbed bolometric luminosity to absorbed count-rate as a function of temperature, redshift, hydrogen column density, and {\it XMM} camera/filter combination. The fake count-rate image is then multiplied by the exposure map of the chosen ODF to account for vignetting, masking and chip gaps. Finally, the fake cluster image is added to the original ODF at the chosen position, and the ODF is run through our source detection/classification pipeline to determine if the fake cluster passes our automated cluster-candidate selection process. For more details on the detection/classification pipeline, refer to Davidson et al. (in preparation).
The process is repeated a total of one hundred times, so that we can build up an average XCS detectability for that parameter combination. Once the full set of $700$ combinations has been tested $100$ times each, the run is complete. We then change the cosmology inputs and start the entire sequence again. The process is very CPU intensive; each selection function run (of $700 \times 100$ combinations) takes several weeks to run on a single node.  For the forecasting work presented herein, we carried out seven selection function runs over the flat $\Lambda$CDM cosmologies with $\Omega_{\rm m} = 0.22, 0.26, 0.28, 0.30, 0.32, 0.34\,{\rm and}\,0.38$. We limit ourselves to flat cosmologies as we use a flatness prior in the forecasting of cosmological constraints.

The resulting selection function is shown in Fig.~\ref{fig:selfun} for the two
luminosity--temperature relations (see Sect.~\ref{lumtempevol}) we consider. Note that the selection function in
regions where we have not calculated it explicitly is extrapolated
from the region where we have done so. Hence, its features in those
extrapolated regions should only be considered a rough indication of
its behaviour, particularly in the high-redshift, high-temperature region. This region is only relevant for including
measurement errors, and since such high-temperature clusters are
exceedingly rare, the uncertainty in this part of the selection
function has no significant impact on our results\footnote{We have subsequently verified the validity of the extrapolation to this level of accuracy with new calculations.}.

\subsubsection{Uncertainty}
\modsec{
The shape of the selection function is dependent on the cluster model employed, as described above.
It is well known that clusters of galaxies have a range of morphologies, with core radii varying from many tens of kpc to a few hundred kpc and $\beta$ values varying generally between $0.45$ and $0.85$ \citep[e.g.][]{RB,OM04,MJFVS08}.}

\modsec{To include the variation of cluster-model parameters in our analysis in a realistic manner, one would require i) a model for the distribution of such parameters among the cluster population (including correlations among parameters), and ii) a characterization of the selection function dependence on such parameters. Lacking either or both of these will produce some level of uncertainty in cluster number predictions and cosmological parameter constraints. However, assessing the level of such uncertainty of course requires a fiducial model (realizing i) and ii)) to compare with. As we do not currently have a realistic model for the model-parameter distribution among the cluster population, it is somewhat premature to carry out such an analysis. In actual data analysis, we intend to model this in detail.  What we can currently do is to compare our standard selection function to one assuming that all clusters have the most extreme values of cluster-model parameters, leading to a gross overestimation of the overall uncertainty in cluster number predictions. We have carried this calculation out for clusters with temperatures typical for the underlying distribution at different redshifts, and describe it below. We again stress that its usefulness for estimating the actual uncertainty in the selection function is limited, as it does not take into account the actual distribution of clusters and their model parameters. Ultimately, we expect that the cosmological constraints we obtain would change little, even if a more realistic model/selection function was used. This is because the changes in the selection function would have a similar impact for the fiducial cosmological model, and for models in its neighbourhood.}

\modsec{We have tested the sensitivity of the selection
function for clusters with $>500$ photons to variations in the cluster core radius $r_{\rm c}$, between the values of $60$~kpc and $260$~kpc (recall the fiducial value used in this work is $160$~kpc, see Sect.~\ref{selfunmodel}). For this we use mock clusters with typical temperatures of $T~=~3\,{\rm keV}$ (and hence luminosities) for a given redshift (as predicted by our models, see Sects.~\ref{Connect}~\&~\ref{Methodology}). In the following, we refer to the relative difference in the selection function detectability, as this is most relevant to the relative difference in numbers of clusters.
Our results show that, for most of the redshift range tested ($0.1\le z\le 1$), clusters with a core radius of $\sim140$ kpc are easier to detect (as extended {\it XMM} sources), than those with smaller or larger core radii. However, the dependence is shallow; the relative uncertainty in the detectability is less than $10$ per cent up to a redshift of $z \sim 0.4$, across the entire $r_{\rm c}$ range. At higher redshifts, the relative uncertainty approaches $30$--$40$ per cent. However, this subset of clusters constitutes only $\sim 30$ per cent of the total population. For higher-temperature clusters, the relative uncertainty drops back to around $10$ per cent at $0.1\le z \le 1$ for $4\,{\rm keV}\le T \le 5\,{\rm keV}$.}

\modsec{In summary, our model for the cluster population is a simplification based on mean observational values of cluster-model parameters, and as a result will have somewhat differing detection properties compared to a real sample. To characterize such uncertainty requires modelling of the cluster-model parameter distribution, and the selection function dependence on those parameters. However, once such information becomes
available, it will be included in the analysis and hence remove/reduce such uncertainty. As we do not currently have
a realistic model of the cluster parameter distribution, we have determined the impact of assuming extreme structural cluster parameters, and found that at most the typical selection function uncertainty is of the order of $10$ per cent.  These results agree with those of \citet{Burenin07} for the 400d survey, which show that reasonable variations in cluster size, morphology and scaling relations induce an uncertainty in the detectability for a given flux of typically less than $5$ per cent.}

\section{From X-ray observables to mass}
\label{Connect}

\subsection{Modeling summary}
Making predictions for X-ray cluster observations requires the modeling of scaling relations to relate temperature to mass, and temperature to luminosity. In addition, the observables will have uncertainties associated with them, which need to be taken into account. The following subsections detail our modeling assumptions, but we summarize them here for reference and orientation.

We first assume that we know a priori exactly how
the cluster X-ray temperature relates to luminosity at the present time, and
how this relation evolves with redshift.  We then study how the constraints on cosmological parameters
degrade if such an assumption is dropped. We consider four extra free
parameters: two parameters to characterize the present-day, power-law,
relation between cluster X-ray temperature and luminosity, another to
describe its redshift evolution as a power of ($1+z$), and lastly one for the
logarithmic dispersion in the (assumed) Gaussian distribution of the intrinsic
(redshift-independent) scatter in the relation between cluster X-ray
temperature and luminosity.

In addition, we evaluate the full impact on the XCS's ability to impose constraints on cosmological parameters that arises from
assuming a dispersion in the Gaussian photometric redshift
distribution of either $5$ or $10$ per cent about the true redshift, both with and
without the presence of unaccounted-for catastrophic errors in the
photometric redshift estimation procedure. Further, we will also
determine the impact of a systematic mis-estimation of the assumed
true dispersion in the photometric redshifts about the true
redshift. Our aim is to test the impact of realistic assumptions
regarding the distribution of photometric redshifts around the true
redshift, and then determine by how much such impact increases by
considering a worst-case scenario.

Similarly, we consider the impact of realistic X-ray temperature
errors obtained from simulations based on the relevant {\it XMM} fields, as
well as significantly larger errors corresponding to a worst-case
scenario. Lastly, we consider the impact of incorrect assumptions about the
cluster scaling relations on cosmological constraints.

Summary tables with our main cluster scaling relation and measurement error assumptions are given in Sect.~\ref{Results}.  Detailed information on the mathematical treatment is given in the Appendix.

\subsection{The X-ray temperature to mass relation}
\label{tempmass}

We need to assume a relation between cluster X-ray temperature and mass to be able to predict cluster distributions. The reason is that presently the effect of cosmological parameters on the galaxy cluster population can only be
accurately predicted as a function of cluster mass \citep[e.g.][]{RB}. The X-ray temperature is one of the best proxy observables, as explained in the Introduction.

\subsubsection{Evolution}  We assume the self-similar prediction
\citep[e.g.][]{Kaiser,BN,Voitb},
\begin{equation}
T\propto M_{\rm v}^{2/3}\,[\Delta_{\rm v}(z)E^{2}(z)]^{1/3}\,,
\end{equation}
for the redshift dependence
of the relation between cluster X-ray temperature and virial mass
to hold for any combination of cosmological parameters, given that
it is consistent with the most recent analyses of observational
data (\citealt{Ettorib,Ettoria}; \citealt*{APP}; \citealt{KVa,KVb,Vik05b,Zhang}).
Here $M_{\rm v}$ is the cluster virial mass, while $\Delta_{\rm
v}(z)$ is the mean overdensity within the cluster virial radius with
respect to the critical density. If the only relevant energy densities
in the Universe are those associated with non-relativistic matter and
a cosmological constant, then
\begin{equation}
E^{2}(z) = \Omega_{\rm m}(1+z)^{3}+\Omega_{k}(1+z)^{2}+\Omega_{\Lambda}\,,
\end{equation}
with $\Omega_{k}=1-\Omega_{\rm m}-\Omega_{\Lambda}$. (Note that we will restrict ourselves to a flat Universe, $\Omega_{k}=0$, in our analysis -- see Sect.~\ref{Cosmology}). Deviations from a self-similar mass--temperature relation will be considered in Sect.~\ref{systbiases}, as explained in the following Section.

\subsubsection{Normalization} The constant of proportionality is set by demanding that for our fiducial
cosmological model (with $\sigma_8 = 0.8$, see Sect.~\ref{Cosmology})
\begin{equation}
M_{500}=3\times10^{14}\,h^{-1}\,{\rm M_{\sun}}\,
\end{equation}
at $z=0.05$ for an X-ray temperature of 5 keV, where $M_{500}$ is
the mass within a sphere centered on the cluster within which its mean
density falls to 500 times the critical density at the cluster
redshift.  In this way, our fiducial cosmological model reproduces the
local abundance of galaxy clusters as given by the HIFLUGCS catalogue
(\citealt{RB}; \citealt*{psw}; \citealt{Viana03}).  Note that such a normalization of the cluster
X-ray temperature to mass relation happens to be very close to that
directly derived from X-ray data by \citet{APP} and \citet{Vik05b}.

The conversion between $M_{500}$ and the halo mass, $M_{180\Omega_{\rm
m}(z)}$, will be carried out by using the formulae derived by
\citet{HK} under the assumption that the halo density profile is of
the NFW type \citep*{NFW95,NFW96,NFW97}, and we will take the
concentration parameter to be 5. This has been shown to provide a good
description of the typical density profile in galaxy clusters (see
\citealt{Arnaud} or \citealt{Voita} and references therein;
\citealt{Vik05b}).

The normalization of the $M$--$T$ relation is subject to a number of uncertainties, the most important of which are the possible violation of hydrostatic equilibrium (\citealt*{RTM04}; \citealt*{NKV07}) and the possible difference between the spectroscopic X-ray temperature and the temperature of the electron gas \citep{Mazzotta04,Rasia05,Vik06}. The precise level of these effects remains to be firmly established, but could be of the order $50\%$ in the normalization mass \citep[e.g.][]{Vik06,NKV07}. The scatter, as well as slope, could also be under-estimated due to these effects \citep{Vik06, NKV07}. We make some estimates of all these systematic effects on cosmological constraints in Sect.~\ref{systbiases}.

\subsubsection{Scatter} We assume that the intrinsic scatter in the relation between
cluster X-ray temperature and mass has a Gaussian distribution
(truncated at $3\sigma$ and re-normalized) with a redshift-independent
dispersion of 0.10 about the logarithm of the temperature. This is motivated
by both cluster X-ray data analysis \citep[e.g.][]{APP,Vik05b,Zhang}
and results from $N$-body hydrodynamic simulations
(e.g. \citealt{Viana03,Borgani,Bal}; \citealt*{KVN}). As explained in the preceding Section, we consider systematic deviations in the scatter in Sect.~\ref{systbiases}.

\subsection{The X-ray luminosity to temperature relation}
\label{lumtemp}

In order to understand how the XCS selection function depends on
cluster mass, we need to know how cluster X-ray luminosity and
temperature relate to cluster mass (see Sect.~\ref{Selfun}).  In
practice, we will use the relation between luminosity
and temperature instead of that between luminosity and
mass, in effect relating these two quantities via the
temperature. This makes sense because the estimation of cluster mass
from X-ray data is always based on the X-ray temperature,
via the assumption of hydrostatic equilibrium, and not on the
luminosity. Thus, while we always need, at least implicitly, to know
how the cluster luminosity relates to temperature to derive the
relation between the luminosity and mass from X-ray
data, the reverse is not true.

As for the mass--temperature relation, assuming self-similarity leads to a specific prediction \citep{Kaiser},
\begin{eqnarray}
L(z,T)  = L(0.05,T)
 \left[\frac{\Delta_{\rm v}(z)E^{2}(z)}
{\Delta_{\rm v}(0.05)E^{2}(0.05)}\right]^{1/2}\,,
\end{eqnarray}
under which clusters with the same X-ray temperature are predicted to
be more X-ray luminous if they have a higher redshift. We have chosen here to normalize the relation with respect to the local ($z=0.05$) relation.
Based on this expression, we write the $L$--$T$ relation in the general form
\begin{eqnarray}
\label{eq:ltrel}
 \log_{10}\left(\frac{L_{\rm X}}{10^{44} h^{-2}\;{\rm erg}\,{\rm
s}^{-1}}\right) = \alpha + \beta \log_{10} \left(\frac{kT}{1\rm
\,keV}\right) + & & \\
\nonumber
 & \!\!\!\!\!\!\!\!\!\!\!\!\!\!\!\!\!\!\!\!\!\!\!\!\!\!\!\!\!\!\!\!\!\!\!\!\!\!\!\!
 \!\!\!\!\!\!\!\!\!\!\!\!\!\!\!\!\!\!\!\!\!\!\!\!\!\!\!\!\!\!\!\!\!\!\!\!\!\!\!\!
 \!\!\!\!\!\!\!\!\!\!\!\!\!\!\!\!\!\!\!\!\!\!\!\!\!\!\!\!\!\!\!\!\!\!\!\!\!\!\!\!\!\!\!\!\!\!\!\!\!\!
 \!\!\!\!
 \gamma_{\rm s} \log_{10} \left[\Delta_{\rm v}(z)E^{2}(z)\right] +
\gamma_z \log_{10}\left(1+z\right) + {\rm N}(0,\sigma_{\log L_{\rm X}})\,. &
\end{eqnarray}
and discuss below the assumptions made for the different parameters.

\subsubsection{Evolution $(\gamma_s, \gamma_z)$}
\label{lumtempevol}
We consider two possible fiducial scenarios, which bracket most
observational results and theoretical expectations: either
\begin{itemize}
  \item no evolution ($\gamma_s=\gamma_z=0$) \,\,\, \textit{or}
  \item self-similar evolution ($\gamma_s=1/2,\gamma_z=0$)
\end{itemize}
for the fiducial combination of cosmological parameters. The parameters $\gamma_s$ and $\gamma_z$ are defined above in equation~(\ref{eq:ltrel}).
  Presently, there is some uncertainty surrounding the redshift evolution of the
relation between cluster X-ray luminosity and temperature.
Essentially, what we know is how that relation behaves for redshifts below 0.3
(e.g. \citealt{Ikebe}; \citealt*{NSH}; \citealt{Ota,Zhang}).
For higher redshifts, the data is still sparse, and the
evidence contradictory, from claims that the relation between
cluster X-ray luminosity and temperature barely evolves at all with redshift
\citep{Holden,Ettorib,Ettoria,Ota,Branchesietal},
to claims that its evolution is close to the self-similar prediction
\citep{NSH,Vik02,Lumb,KVa,Mau,Zhang,Hicks:2007ip}. Some authors argue that self-similarity remains viable at all redshifts, and that at least some of the observed discrepancies could be due to selection effects, as the Malmquist bias from scaling-relation scatter (also discussed below) could distort the deduced evolution if the sample selection is not sufficiently understood \citep[e.g.][]{Branchesietal,MauYX,XMMLSSfirst,Nord}. On the other hand,  \citet{hilton2215} argue for deviation from the self-similar prediction based on a set of high-redshift clusters combined with the recently discovered XCS cluster XMMXCS~J2215.9-1738 at $z = 1.457$.

When the XCS catalogue becomes available, the relation between
cluster X-ray luminosity and temperature, as a function of redshift,
will be estimated jointly with the cosmological parameters, but for now we will have to rely on the limited information available.

\subsubsection{Normalization \& slope $(\alpha, \beta)$}  We assume the local ($z=0.05$) relation between the cluster X-ray
luminosity in the {\it ROSAT} $[0.1,2.4]$ keV band and temperature to be
\begin{equation}
\log_{10}\left(\frac{L_{\rm X}}{h^{-2}\;{\rm erg}\,{\rm s}^{-1}}\right)=
42.1+2.5\log_{10}\left(\frac{kT}{1\,{\rm keV}}\right)\,,
\end{equation}
as was derived in \citet{Viana03} for a combination of cosmological
parameters similar to those assumed for our fiducial cosmological
model. The X-ray data used in \citet{Viana03} was that of galaxy clusters present in the HIFLUGCS catalogue \citep{RB}, and
therefore the conversion between $L_{\rm X}$ and X-ray bolometric
luminosity is performed through a fit (derived by us) based on the
values both quantities take for the galaxy clusters in HIFLUGCS,
\begin{equation}
L_{\rm bol} = \frac{L_{\rm X}}{0.25 + 0.7\exp\left(-0.23kT/1\,
{\rm keV}\right) \,.}
\end{equation}
As in \citet{Ikebe}, the relation between the cluster X-ray luminosity
and temperature derived in \citet{Viana03} takes into account the fact
that any flux-limited sample of galaxy clusters will be composed of
objects which are on average more X-ray luminous than the mean
luminosity of all existing galaxy clusters with the same redshift and
X-ray temperature. This Malmquist type of bias increases with
decreasing temperature, and thus ignoring it leads not only to an
overestimation of the normalization of the relation between luminosity and temperature,  but also to an underestimation of its
slope.

\subsubsection{Scatter $(\sigma_{\log L_{\rm X}})$}
\label{ltscatter}
We assume that the intrinsic scatter in the relation between
cluster X-ray luminosity (in the 0.1 to 2.4 keV band) and temperature
has a redshift-independent Gaussian distribution (truncated at
$3\sigma$ and re-normalized) about the logarithm of the X-ray
luminosity, with $1\sigma$ dispersion $\sigma_{\log L_{\rm X}} = 0.30$
\citep{Ikebe,Viana03}. This is also close to what was found by
\citet{Kay07} in the CLEF simulation. Although Kay et al. also observe an
evolution of the scatter with redshift, there is no strong
observational evidence for or against such an evolution at present, and therefore we do not include it in our analysis.

The existence of intrinsic scatter in the relation between cluster
luminosity and mass will effectively increase the observed
number of galaxy clusters above any X-ray luminosity (or flux)
threshold, relative to the case without scatter. This results from the steepness of the cluster mass function,
due to which significantly more clusters have their X-ray luminosity
scattered up than down across any given luminosity
threshold. Therefore, intrinsic scatter between
X-ray luminosity and mass can have a
considerable impact on the predicted number of XCS clusters and on the
estimation of the constraints the XCS will impose on
cosmological parameters. This scatter can be considered as the combination of the scatter in the luminosity to
temperature and temperature to mass relations, with clear
observational evidence that the former dominates over the latter
\citep{stanekl,Zhang}.

As higher redshifts are considered, it is expected that an increasing
number of galaxy clusters will have undergone recent major mergers,
not only leading to increased scatter in the cluster scaling relations
but also making its distribution highly non-Gaussian, with long tails
developing towards both high X-ray luminosity and, to a lesser degree,
temperature, at fixed mass \citep*{RSR}. This has the potential to
substantially affect the estimation of the constraints the XCS will be
able to impose on cosmological parameters. There is a lack of high-redshift observational data in this regard and we are also not confident that we will detect, for the purposes of understanding this behaviour, many useful clusters at $z>1$. We therefore chose to consider in the estimation procedure only those clusters in the mock XCS catalogues which have a redshift $z \le 1$.

\subsection{Photometric redshifts}
\label{Photozs}

\subsubsection{The role of photometric redshifts}

Redshifts are required for XCS clusters, both to place them correctly in the evolutionary sequence and to allow the measurement of X-ray temperatures from {\it XMM} spectra. With regard to the latter point, pure thermal bremsstrahlung spectra are essentially featureless (barring a high-energy cut-off), making them degenerate in temperature and redshift. Therefore, in the absence of independent redshift information, all one can measure from a typical XCS cluster spectrum would be a so-called apparent X-ray temperature, i.e. one scaled by  $(1+z)$, see Appendix~\ref{clustcnterr}. As shown by \citet{LVRM}, these apparent temperatures are not sufficient to allow one to measure cosmological parameters from cluster catalogues. As a result, optically-determined redshifts will be required for almost all clusters in the XCS catalogue (the exception being a tiny number that are detected with sufficient signal to noise to allow X-ray emission features, such as the Iron K complex at $\sim 7$ keV, to be resolved over the thermal continuum).

As is now typical for cluster surveys \citep[e.g.][]{GY05}, the XCS is relying heavily on the photometric redshift technique for its optical follow-up. This is because photometric redshifts are much more efficient, in terms of telescope time requirements, than spectroscopic redshifts. However, they have the disadvantage, over spectroscopic redshifts, that the redshift errors are larger and sometimes poorly understood. The XCS is using both public-domain photometry (e.g. from SDSS and 2MASS) and proprietary data from the NOAO--XCS Survey (NXS, \citealp{nxs}) to both optically confirm (as clusters) XCS candidates and to measure photometric redshifts. To date, more than $400$ XCS candidates have been optically confirmed in this way.

Errors on photometric redshifts must be accounted for when determining cosmological parameters from cluster surveys, and so we have included prescriptions for such errors in the forecasting work presented herein. Our prescriptions include both purely statistical errors and so-called catastrophic systematic errors. As shown by previous work \citep{Hut04,Hut06,LimaHuphotoz}, purely statistical errors have a negligible impact on cosmological parameter constraints. By contrast, if catastrophic errors are not accounted for properly in the fitting, they could have a significant impact on cosmological parameter constraints.  We note that previous work has concentrated only on the impact of redshift errors on the evolutionary sequence, whereas we have also included the impact of photometric errors on X-ray temperature determinations.

\begin{figure}
\includegraphics[width=0.95\linewidth]{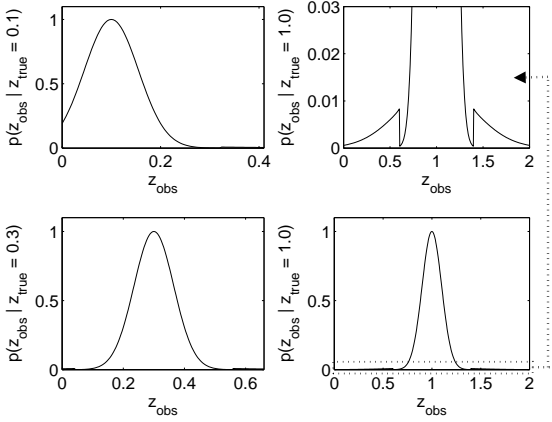}
\caption{Realistic redshift error distributions at various redshifts. The upper right panel shows a magnification of the bottom-right distribution, highlighting the catastrophic-error part of the distribution.}
\label{fig:zerrdists}
\end{figure}

\subsubsection{Distribution} Following \citet{Hut04} we assume that the statistical error in the
photometric redshifts of individual galaxy clusters has a Gaussian
distribution about the true redshift, $z_{\rm t}$. In an attempt to
reproduce the expected degradation with redshift of the absolute
accuracy of cluster photometric redshift estimation methods, and in contrast to \citet{Hut04} (but in the same way as \citealt{LimaHuphotoz}), we will
assume the dispersion to be proportional to $(1+z_{\rm t})$,
normalized at $z_{\rm t}=0$ to either $\sigma_0 = 0.05$ or $\sigma_0 = 0.10$
\citep{GY00,Gladders,GY05,Getal06}.  Unaccounted-for systematic errors
in the photometric redshift estimation procedure are much harder to
model, because they can take a variety of guises.  We will consider
here one such type of error: catastrophic errors in the photometric
redshift estimation procedure.
The existence of unaccounted-for catastrophic errors will be modelled
by assigning a random photometric redshift error to either a fraction $f_{\rm cat} = 0.05$ or $f_{\rm cat}= 0.10$ of the galaxy clusters, drawn
from a Gaussian distribution that has four times the dispersion of the
standard distribution, with the requirement that the photometric
redshift error has to be at least $1\sigma$ away from the true
redshift. The functional form of the redshift error distribution is
given in Appendix~\ref{clustcnt}.

We label the case $\{\sigma_0=0.05, f_{\rm cat}=0.05\}$
`realistic' and the case $\{\sigma_0=0.10, f_{\rm cat}=0.10\}$
`worst-case' redshift errors. Examples of realistic redshift error
distributions are shown in Fig.~\ref{fig:zerrdists}.

\subsection{X-ray temperature}
\label{Terror}

\subsubsection{Estimating the measurement errors}

Initial estimates \citep{LVRM} showed that X-ray temperatures measured for XCS clusters are expected to have an associated measurement uncertainty of less than $20$ per cent at 1$\sigma$. However, these estimates were based on a photon count of $1000$ and assume a single hydrogen column density over the {\it XMM} fields, and are therefore not directly applicable to our $^{500}$XCS sample.  Hence, in order to estimate the temperature errors that will be present in the XCS statistical sample more accurately, we have conducted Monte Carlo simulations using the X{\sc spec} spectral fitting package
\citep{arnaud96a}. We created $1000$ sets of fake
spectra for the {\it XMM--Newton} EPIC PN and MOS instruments, from a
{\sc mekal} plasma model \citep*{mewe86a} multiplied by a {\sc wabs}
photo-electric absorption model \citep{morrison83a}. Responses for a
mean off-axis angle were used and a mean background was added. The
model was then fitted to each of the spectra to derive a
temperature. A plasma metallicity of $0.3{\rm Z_{\sun}}$ was used throughout, in accordance with the assumptions in our selection function calculations (see
Sect.~\ref{Selfun}), and we assume a photon count of $500$. This procedure was repeated for a range of input temperatures, redshifts and absorption columns. We then marginalize over the hydrogen absorption columns using the expected hydrogen column distribution for our {\it XMM} fields.

\begin{figure}
\includegraphics[width=8.5cm]{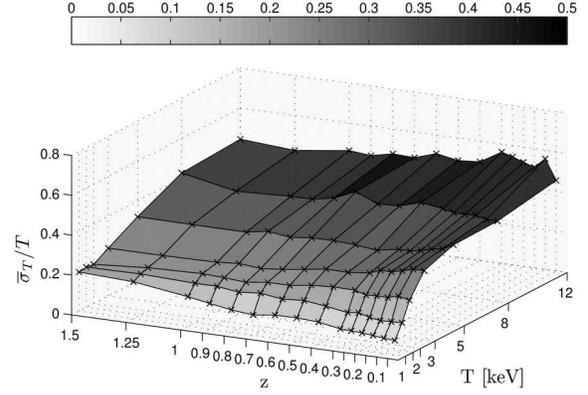}
\caption{Mean fractional temperature errors from the simulations
performed, for $500$ photons, and as marginalized over expected absorption columns for the
XCS.} \label{fig:terrdisted}
\end{figure}

\begin{figure}
\includegraphics[width=0.95\linewidth]{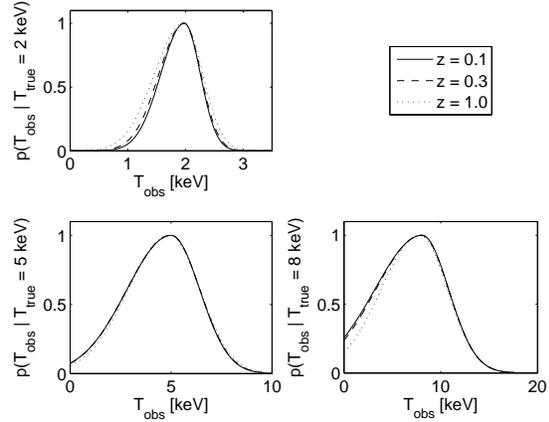}
\caption{Realistic temperature error distributions at various
redshifts and temperatures, based on our {\it XMM--Newton} simulations, for a photon count of $500$.} \label{fig:terrdists}
\end{figure}

The mean fractional temperature errors from our simulations are shown in
Fig.~\ref{fig:terrdisted}.
The largest influence on the temperature errors comes from the
input temperature itself. Since metal lines in the spectrum provide
much better constraints on the temperature than the shape of the
bremsstrahlung continuum, and the fraction of line emission in the
spectrum declines with increasing temperature, the errors are
larger for hotter clusters. The effect of redshift on the errors is
much smaller and itself temperature-dependent. For low-temperature
systems at high redshifts, part of the X-ray spectrum is shifted out of the bottom of
the {\it XMM} passband, increasing the errors. For high-temperature systems, the effect of increasing redshift is to shift the source spectrum to lower energies for which the {\it XMM} effective area is larger, thus decreasing the errors.

\subsubsection{Distribution}
The distribution of temperatures obtained in our simulations was fitted by an
asymmetric Gaussian function to parametrize the temperature error
distribution, with the fractional error given by a two-dimensional quadratic expression in temperature and redshift.  We marginalize over the distribution of
absorption columns found in XCS fields to obtain mean parameters for
our asymmetric Gaussian error distribution.  The exact functional form of the fitted error distribution is given in equation~(\ref{eq:tdist}) in
Appendix~\ref{clustcnt}.

We will label the case with $\sigma_{T}$ according to our simulation
results as `realistic' and the case with three times this dispersion
as `worst-case' temperature errors. Examples of realistic temperature
error distributions are shown in Fig.~\ref{fig:terrdists}.  Note that, as we are assuming that all detected clusters have a photon count of exactly $500$, our error distributions represent a worst-case scenario in this regard.

\section{From mass to cosmology and constraints}
\label{Methodology}

\subsection{The mass function}
\label{Massfunction}
Having connected our direct X-ray observables to cluster mass using the methodology in the preceding Section, we can then combine these relations with the mass function (below) to find the cluster distribution as a function of temperature and redshift.
The differential number density of haloes in a mass
interval ${\rm d}M$ about $M$ at redshift $z$ can be written as
\begin{equation}
n(M,z)\,{\rm d}M = -F(\sigma)\,\frac{\rho_{\rm m}(z)}{M\sigma(M,z)}\,
\frac{{\rm d}\sigma(M,z)}{{\rm d}M}\,{\rm d}M\,,
\end{equation}
where $\sigma(M,z)$ is the dispersion of the density field at some
comoving scale $R=(3M/4\pi\rho_{\rm m})^{1/3}$ and redshift $z$, and
$\rho_{\rm m}(z) = \rho_{\rm m}(z=0)(1+z)^3$ the matter density.

\subsubsection{Parametrization}
It has been shown by \cite{jetal} that a good fit (accurate to better
than 20 per cent) to the mass functions recovered from various large
$N$-body simulations, in the regime $-1.2\leq\ln\sigma^{-1}\leq1.05$,
is given by
\begin{equation}
F_{\rm J}(\sigma) = 0.315 \, \exp\left[-|\ln\sigma^{-1}+
0.61\,|^{3.8}\right]\,,
\end{equation}
where the halo mass is defined at a mean
overdensity of 180 with respect to the background matter density,
independently of the cosmological parameters assumed, or equivalently
to a mean overdensity of $180\Omega_{\rm m}(z)$ with respect to the
critical density. This result has been confirmed by
\citet{Evrard02,HK,Klypin,LJ,Reed03}; \citet*{LBH}; \citet{Warren,Reed07}, and we will
therefore use this fit to estimate the expected comoving
number density of haloes for any given combination of cosmological
parameters. This also makes the like-for-like comparison with other
cluster constraints straightforward, as most rely on the Jenkins mass
function. The dispersion $\sigma$ is calculated using a fit in analogue to \citet{VL99}, which is accurate to within two percent for the range of halo masses relevant for this work, compared to the exact expression employing the BBKS transfer function \citep{BBKS}. \modsec{(In a real data analysis, this prescription would not be sufficiently accurate, however for forecasting purposes it is acceptable.)}

\subsection{Cosmology}
\label{Cosmology}
We have already seen that cosmology enters into the prediction of cluster numbers as a function of temperature and redshift through the selection function, the cluster scaling relations, and the mass dispersion. Additionally, the cosmic volume $dV/dz$ will also enter as we need to multiply the differential distribution by this quantity (discussed in the following Section).

\begin{figure*}
\centering
\includegraphics[width=13cm]{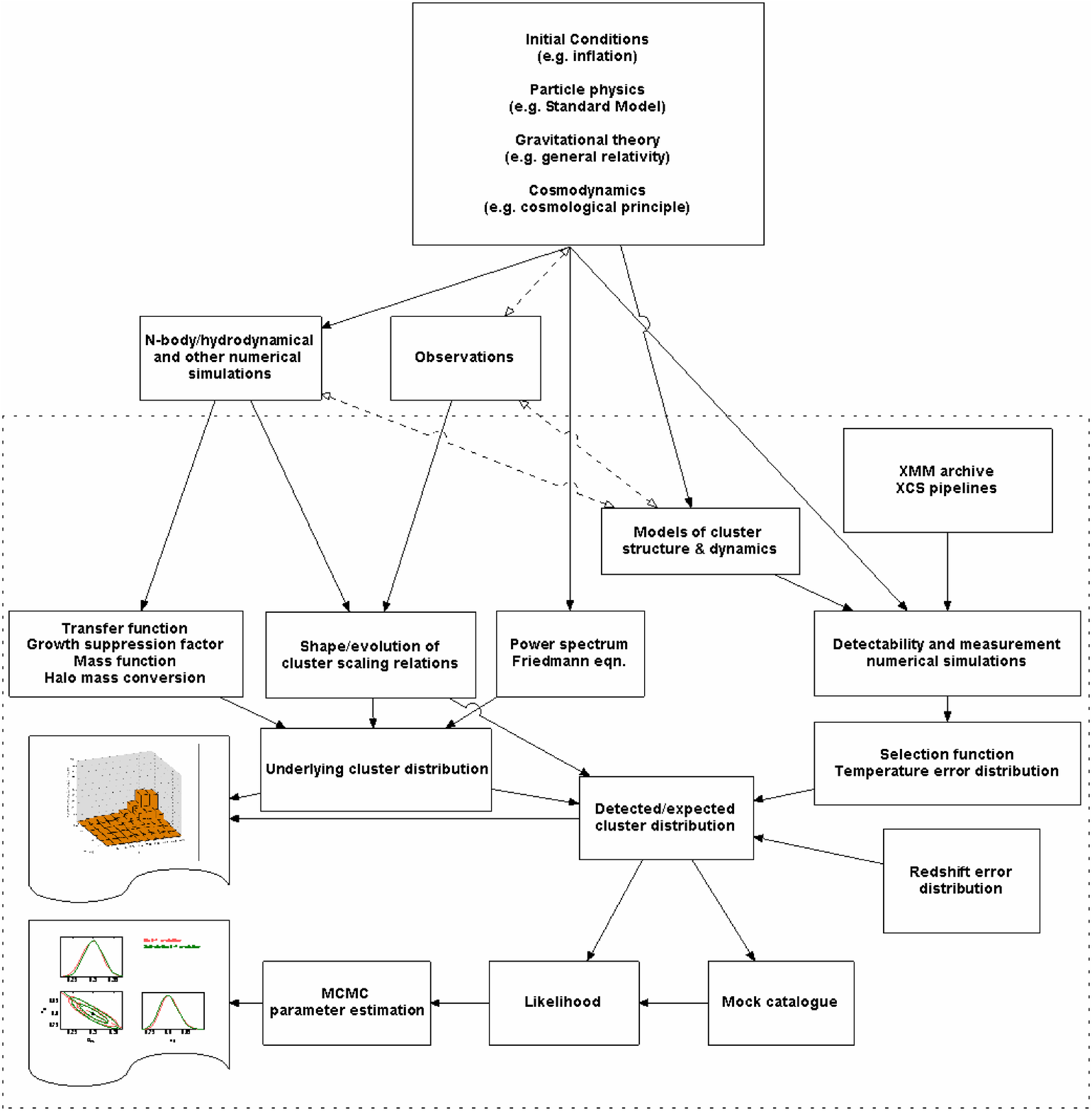}
\caption{Flowchart for cluster predictions and forecast parameter
estimation. The dash-enclosed area indicates the processes that enter
in our calculations. Bi-directional dashed arrows are used to indicate the main circular relations, where information from one part is used to inform another, which then informs the first, and so on.} \label{fig:flowchart}
\end{figure*}

\subsubsection{Parameters} We work within the Cold Dark Matter (CDM) paradigm, with a
spectrum of primordial adiabatic Gaussian density perturbations.
We assume that $\Omega_{\rm m}=0.3, \Omega_{\Lambda}=0.7$, $\sigma_8
= 0.8$, $\Omega_{\rm b} = 0.044$ and $h=0.75$. As we do not
expect the XCS to have particularly competitive constraining power on
$\Omega_k$, we restrict our analysis to the case of a flat universe,
$\Omega_k=0$, in accordance with observations of e.g. the cosmic
microwave background \citep{wmap06}. We take
the present-day shape of the matter power spectrum to be well
approximated by a CDM model with scale-invariant primordial
density perturbations whose transfer function shape parameter is
$\Gamma \approx \Omega_{\rm m}h\times\exp[-\Omega_{\rm
b}(1+\sqrt{2h}/\Omega_{\rm m})] = 0.18$. This is the mean value
obtained from different analyses of SDSS data
\citep{szetal,Pope,tegmark,SDSS05,Blake,padman06} and also perfectly
compatible with the 3-year WMAP data \citep{wmap06}.  We have checked
that assuming $\Gamma$ is either 0.16 or 0.20 does not change our
results. (In a real data analysis, using the shape parameter is too
simplistic, but for forecasting purposes it is sufficient.)  A summary of our cosmological parameter assumptions is given in Table~\ref{tab:cosmology}.

\begin{table}
\centering
\begin{tabular}{|l|c|c|}
  \hline
  Parameter & Value & Prior \\
  \hline
  $\Omega_{\rm m}$ & $0.3$ & $[0.1,1]$ \\
  $\Omega_{\Lambda}$ &  $0.7$ & $1-\Omega_{\rm m}$ \\
  $\sigma_8$ & $0.8$ & $[0.3,1.3]$ \\
  $\Omega_{\rm b}$ & $0.044$ & $0.044$ \\
  $h$ & $0.75$ & $0.75$ \\
  $n_{\rm s}$ & $1$ & $1$ \\
  \hline
\end{tabular}
\caption{Cosmology assumptions used. Fiducial values are given first, followed by priors assumed in parameter estimation.}
\label{tab:cosmology}
\end{table}

\subsection{Combining observables and cosmology}
As we have seen, our cluster distribution calculations involve many different steps and components. Importantly, they rely on both simulation and observational data, as well as direct theoretical input. A schematic overview of the relevant inputs, processes and outputs is shown in a flowchart form in Fig.~\ref{fig:flowchart}.
Collecting all components, the number of clusters in ${\rm d}T {\rm d}z$ around $(T,z)$ is given by
\begin{eqnarray}
  n\left(M(T,z), z\right)\frac{{\rm d}M}{{\rm d}T} f_{\rm sky}(L(T,z), T, z) \frac{{\rm d}V}{{\rm d}z}
   {\rm d}T {\rm d}z
\end{eqnarray}
where $f_{\rm sky}$ combines survey area and selection function. This expression ignores scatter in the scaling relations and measurement errors. A complete treatment is given in Appendix~\ref{clustcnt}.  The remaining component for arriving at parameter constraints is the likelihood, which is described next.

\begin{table*}
\centering
\begin{tabular}{|l|c|c|c|c|c|}
   \hline
   Parameter & Description & \includegraphics[width=2.6cm]{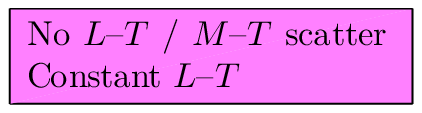}
   & \includegraphics[width=2.6cm]{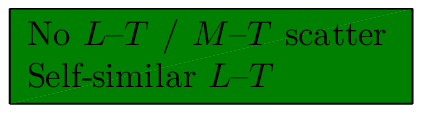} & \includegraphics[width=2.6cm]{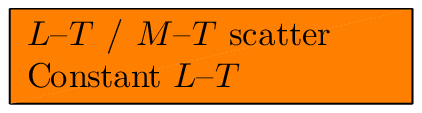} &  \includegraphics[width=2.6cm]{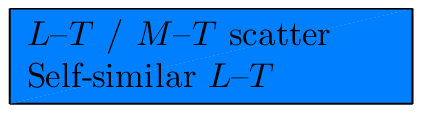} \\
  \hline
  Colouring & & Pink & Green & Orange & Blue \\
  \hline
  $L$--$T$: \cr $\alpha$ & Normalization & \begin{tabular}{c}
                                             $-1.90$ \\
                                             $[-1.90]$
                                           \end{tabular}
   & \begin{tabular}{c}
                                             $-1.92$ \\
                                             $[-1.92]$
                                           \end{tabular}
   & \begin{tabular}{c}
                                             $-1.90$ \\
                                             flat, unrestricted
                                           \end{tabular}
                                           & \begin{tabular}{c}
                                             $-1.92$ \\
                                             flat, unrestricted
                                           \end{tabular}
    \\
    \\
  $\beta$ & Slope & \begin{tabular}{c}
                                             $2.5$ \\
                                             $[2.5]$
                                           \end{tabular}
  & \begin{tabular}{c}
                                             $2.5$ \\
                                             $[2.5]$
                                           \end{tabular} & \begin{tabular}{c}
                                             $2.5$ \\
                                             flat, unrestricted
                                           \end{tabular}
                                           & \begin{tabular}{c}
                                             $2.5$ \\
                                             flat, unrestricted
                                           \end{tabular} \\ \\
  $\gamma_s$ & Self-similarity exp. & \begin{tabular}{c}
                                             $0$ \\
                                             $[0]$
                                           \end{tabular}
  & \begin{tabular}{c}
                                             $1/2$ \\
                                             $[1/2]$
                                           \end{tabular} & \begin{tabular}{c}
                                             $0$ \\
                                             $[0]$
                                           \end{tabular} & \begin{tabular}{c}
                                             $1/2$ \\
                                             $[1/2]$
                                           \end{tabular} \\ \\
  $\gamma_z$ & Redshift exp. & \begin{tabular}{c}
                                             $0$ \\
                                             $[0]$
                                           \end{tabular} & \begin{tabular}{c}
                                             $0$ \\
                                             $[0]$
                                           \end{tabular} &
                                           \begin{tabular}{c}
                                             $0$ \\
                                             $[-1,1.5]$
                                           \end{tabular} &
                                           \begin{tabular}{c}
                                             $0$ \\
                                            $[-1,1.5]$
                                           \end{tabular} \\ \\
  $\sigma_{\log L_{\rm X}}$ & Scatter & \begin{tabular}{c}
                                             $0$ \\
                                             $[0]$
                                           \end{tabular}
  & \begin{tabular}{c}
                                             $0$ \\
                                             $[0]$
                                           \end{tabular} & \begin{tabular}{c}
                                             $0.3$ \\
                                             $[0.2,0.4]$
                                           \end{tabular}
                                           & \begin{tabular}{c}
                                             $0.3$ \\
                                             $[0.2,0.4]$
                                           \end{tabular} \\ \\
  $m_L$ & Max. scatter in units of $\sigma_{\log L_{\rm X}}$ &
  \begin{tabular}{c}
                                             --
                                           \end{tabular} &
                                           \begin{tabular}{c}
                                             --
                                           \end{tabular} & \begin{tabular}{c}
                                             $3$ \\
                                             $[3]$
                                           \end{tabular}
  & \begin{tabular}{c}
                                             $3$ \\
                                             $[3]$
                                           \end{tabular} \\
  \hline
  $M$--$T$: \cr evolution & & \multicolumn{4}{c}{self-similar,
   normalized to HIFLUGCS} \\ \\
  $\sigma_{\log T}$ & Scatter &
  \begin{tabular}{c}
                                             $0$ \\
                                             $[0]$
                                           \end{tabular} &
                                           \begin{tabular}{c}
                                             $0$ \\
                                             $[0]$
                                           \end{tabular} &
                                           \begin{tabular}{c}
                                             $0.1$ \\
                                             $[0.1]$
                                           \end{tabular} &
                                           \begin{tabular}{c}
                                             $0.1$ \\
                                             $[0.1]$
                                           \end{tabular} \\ \\
  $m_T$ & Max. scatter in units of $\sigma_{\log T}$ & \begin{tabular}{c}
                                             --
                                           \end{tabular}
  & \begin{tabular}{c}
                                             --
                                           \end{tabular} &
                                           \begin{tabular}{c}
                                             $3$ \\
                                             $[3]$
                                           \end{tabular} & \begin{tabular}{c}
                                             $3$ \\
                                             $[3]$
                                           \end{tabular} \\
  \hline
\end{tabular}
\caption{Cluster scaling relation assumptions and their labelling. Fiducial values are given first, followed by priors assumed in parameter estimation below (usually in brackets). Note that the colour coding at the top of the table is used to indicate these fiducial models throughout. See Sects.~\ref{tempmass}~\&~\ref{lumtemp} and the Appendix for details.}
\label{tab:scalrel}
\end{table*}

\subsection{Likelihood}
Turning our attention to using the cluster distribution prediction for cosmological constraints and forecasting, we need an expression for the likelihood of an observed catalogue of galaxy clusters. The likelihood ${\cal L}$ for a given observed catalogue is simply the
product of the Poisson probabilities of observing $N_i$ XCS clusters
in the bin with widths $\Delta T, \Delta z$ centered at each of the
$(T_i,z_i)$ positions,
\begin{eqnarray}
{\cal L}\, &=& \, \prod_{i}\left[\frac{\lambda_i^{N_i}}{N_i!}  \,
e^{-\lambda_i}\right]\,
\end{eqnarray}
where
\begin{equation}
\lambda_i = N(T_i-\Delta T/2,
T_i+\Delta T/2, z_i-\Delta z/2, z_i+\Delta z/2)
\end{equation}
is the expected number of XCS clusters in bin $i$, taking into account
sky coverage, survey selection function, and any uncertainties in
scaling relations or measurements (see equations in Appendix
\ref{clustcnt}).  We do not take into account the fact that the
positions of galaxy clusters are spatially correlated, because the
mean distance between XCS clusters is typically much larger than the
observationally determined correlation length in the range
$10-20\,h^{-1}\,{\rm Mpc}$ (see e.g. \citealt{Nicholcorrel,Romercorrel,collinscorrel}; \citealt*{Gonzalezcorrel}; \citealt{Brodwincorrel}), as a result of the {\it XMM} pointings
being scattered all over the sky.
Even if the XCS area was contiguous, given the
very large depth of the XCS, the impact of cluster spatial
correlations on the estimation of cosmological parameters with the XCS
galaxy cluster abundance data would be small
\citep{WhiteM,HK,Holder,HC}.

\begin{figure}
\includegraphics[width=8.5cm]{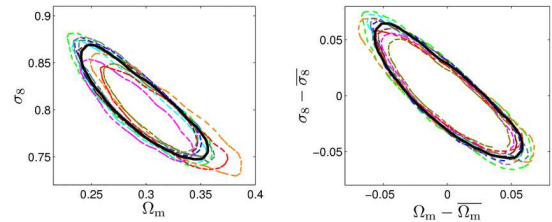}
\caption{Parameter constraints ($95$ per cent confidence level) for a set of 10 random realizations of
the catalogue Poisson distribution (dashed coloured lines) compared to
the average-catalogue parameter constraint (solid black line). In the
right-hand panel, each contour has been re-centered around its
distribution mean. A constant $L$--$T$ relation and no $L$--$T$ or
$M$--$T$ scatter was assumed.} \label{fig:contoursrand}
\end{figure}

\begin{table}
\begin{tabular}{|c|c|}
  \hline
  Quantity & Labels/assumptions \\
  \hline
  Redshift $z$ & \includegraphics[width=5.5cm]{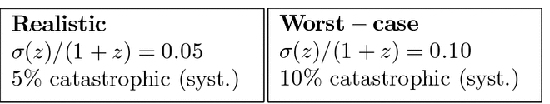} \\
  X-ray temperature $T$ & \includegraphics[width=5.5cm]{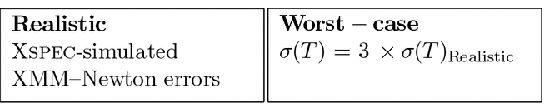} \\
  \hline
\end{tabular}
\caption{Summary of measurement error assumptions and their labelling. See Sects.~\ref{Photozs}~\&~\ref{Terror} and the Appendix for details.}
\label{tab:errscheme}
\end{table}

As we are seeking to obtain expected/typical constraints, in a sense
a Maximum Likelihood (ML) point estimate, we use $N_i = \lambda_i^*$, where the asterisk
denotes fiducial-model values. Using this `average-catalogue' construction, we obtain an
excellent estimate of the size and shape of the expected likelihood contours, but
avoid the offset in the best fit away from the fiducial parameter
values that is associated with, e.g., the most likely catalogue. Any random realization of a Poisson sample will exhibit such an offset. Examples can be seen in the right-hand panel of Fig.~\ref{fig:contoursrand}, where the results for the average-catalogue method is compared to those for random catalogue realizations. We wish to
avoid offsets of this type as we are mainly interested in the shape and size
of contours, or wish to separate possible biases from such an
offset. This methodology is explained and motivated in detail in
Appendix~\ref{explik}. As stated above, Fig.~\ref{fig:contoursrand} compares
constraints derived using this method to constraints derived from a
Poisson sample of mock catalogues. The results confirm that
constraints derived using our methodology provide an excellent
estimate of the expected constraints. Note that in future real data
analyses this methodology cannot be used, and there will in
general be some offset.

The exploration of the likelihood function in parameter space was
carried out using a custom code based on standard Monte Carlo Markov Chain techniques
(\citealt{GelmanRubin}; \citealt*{Gilks}; \citealt{Lewis:2002ah,Verde:2003ey,Tegmark:2003ud,
Dunkley:2004sv}). The calculation of the integrals involved in the
likelihood was done with the state-of-the-art numerical integration
packages \textsc{CUBPACK} \citep{CUBPACK} and \textsc{Cuba} \citep{CUBA}.

\begin{figure*}
\centering
\subfloat[Underlying cluster distribution. Note that only the $M$--$T$ relation is relevant for the underlying distribution, and we therefore colour according to both $L$--$T$ assumptions with the same $M$--$T$ relation.]{\label{fig:2dhistsnosel}
\includegraphics[width=9.3cm]{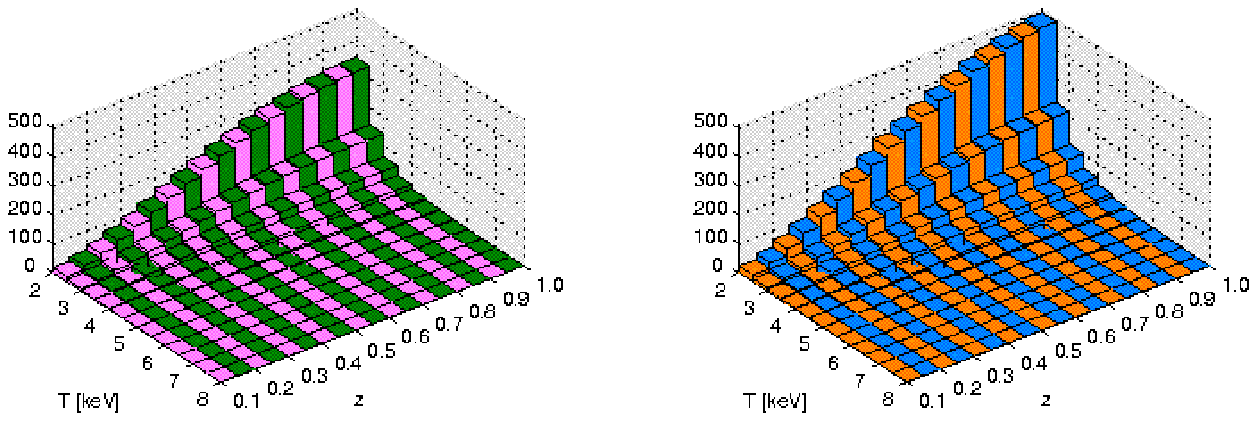}} \\
\subfloat[Expected detections using selection
function.]{\label{fig:2dhists}
\includegraphics[width=9.3cm]{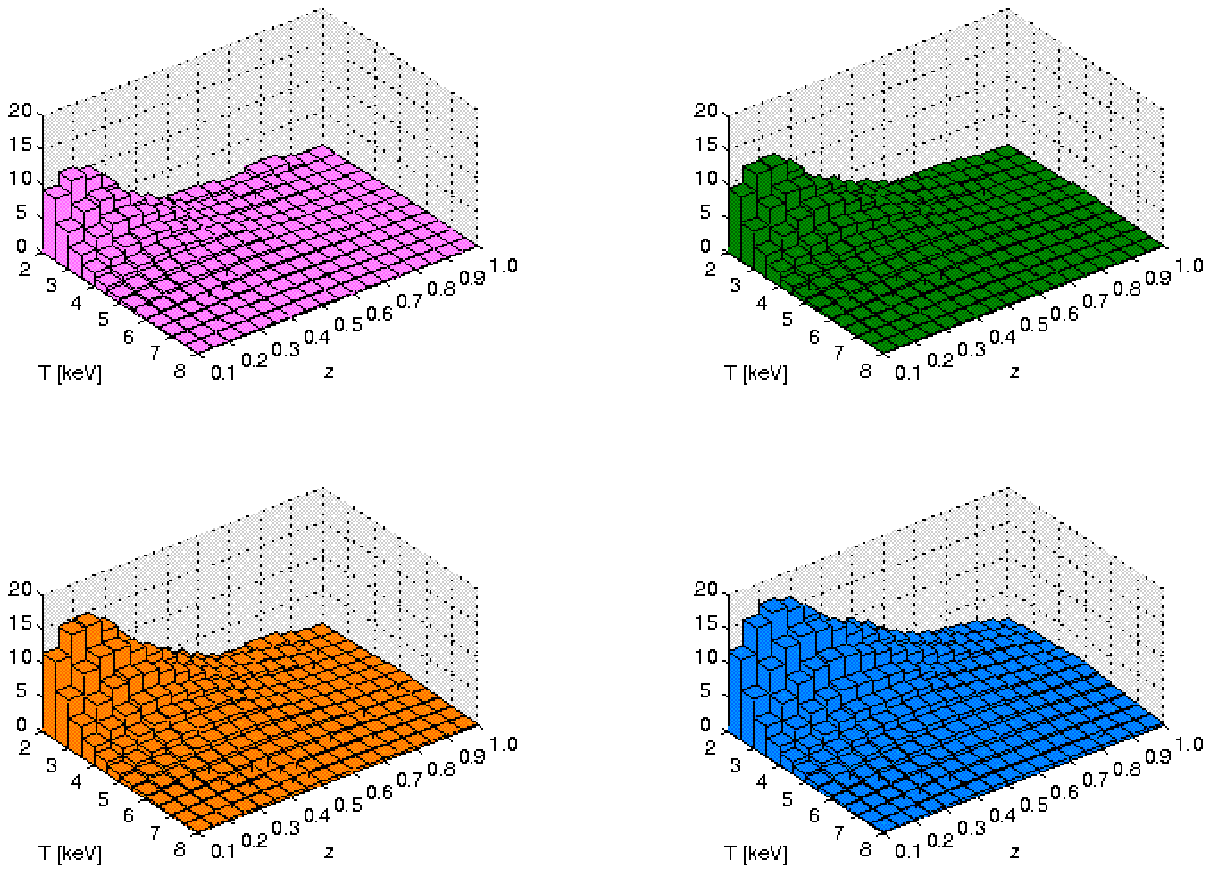}}
\subfloat[Detected fraction per bin.]{\label{fig:2dhistscf}
\includegraphics[width=9.3cm]{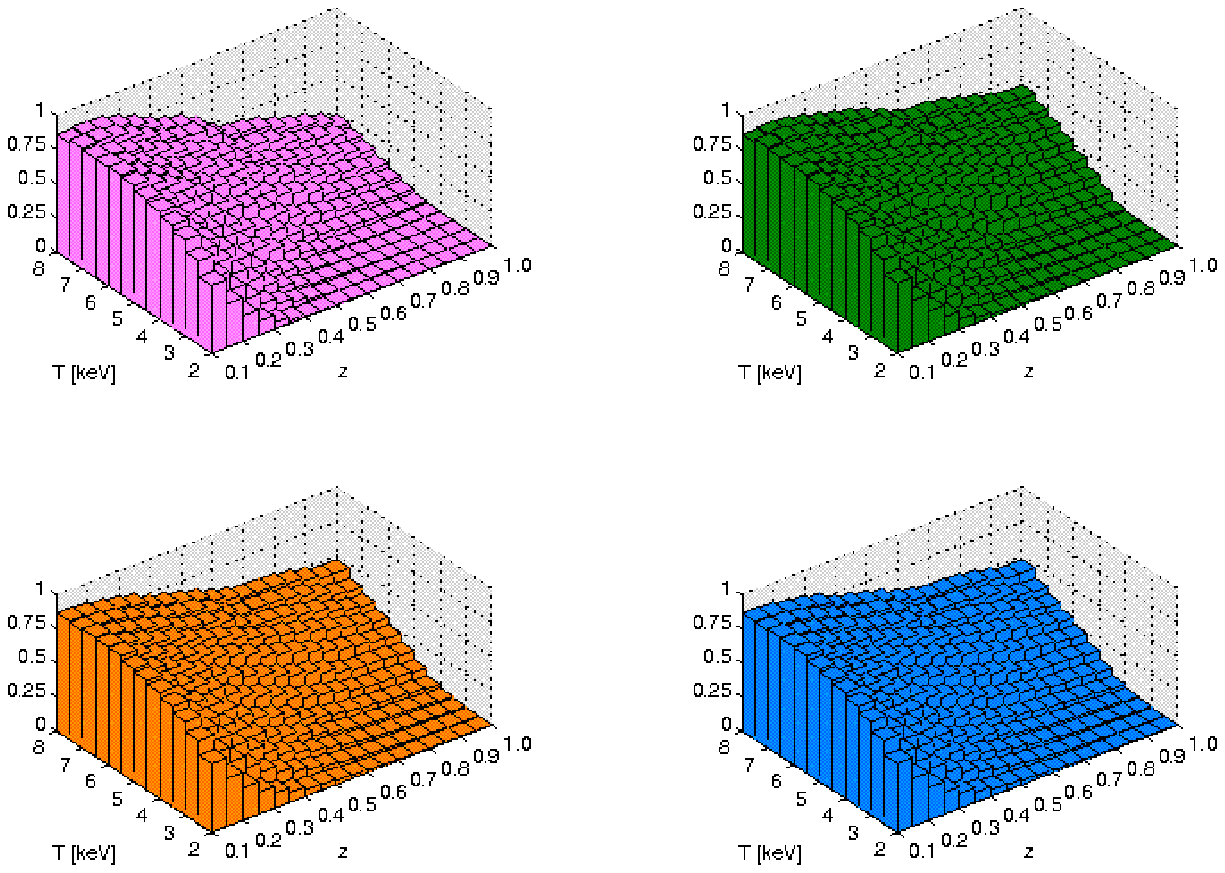}}
\caption{Expected cluster number count distributions for $^{500}$XCS, for no
$L$--$T$ nor $M$--$T$ scatter and no $L$--$T$ evolution (pink), no
$L$--$T$ nor $M$--$T$ scatter and self-similar $L$--$T$ evolution
(green), $L$--$T$ and $M$--$T$ scatter and no $L$--$T$ evolution
(orange), and $L$--$T$ and $M$--$T$ scatter and self-similar $L$--$T$
evolution (blue). Bin sizes are $\Delta z = 0.05$ and $\Delta T = 0.5\,{\rm keV}$.}
\end{figure*}

\begin{figure}
\centering
\subfloat[Underlying cluster
distribution. Note that only the $M$--$T$ relation is relevant for the underlying distribution, and we therefore colour according to both $L$--$T$ assumptions with the same $M$--$T$ relation.]{\label{fig:histogramsnosel}
\includegraphics[width=\linewidth]{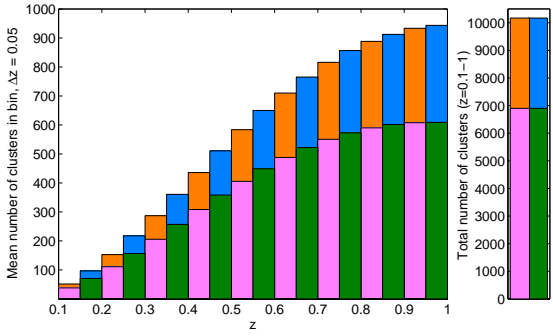}} \\
\subfloat[Expected detections using selection
function.]{\label{fig:histograms}
\includegraphics[width=\linewidth]{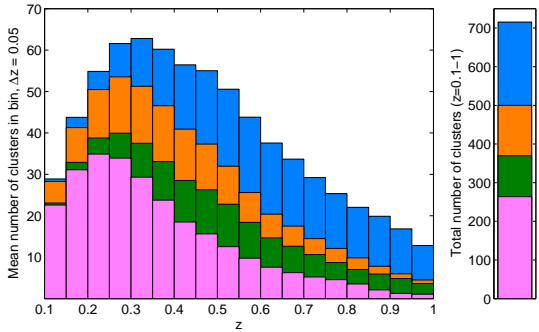}} \\
\subfloat[Detected fraction of clusters per
bin.]{\label{fig:histogramscf}
\includegraphics[width=\linewidth]{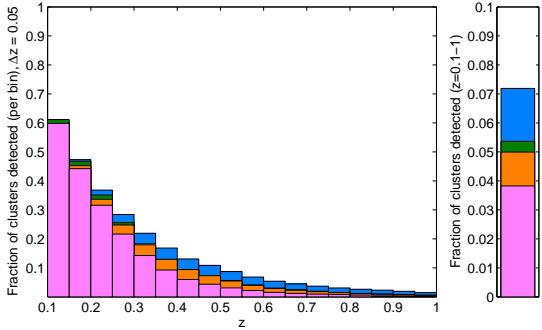}}
\caption{Expected cluster distributions for the $^{500}$XCS, for our
four different cluster scaling relation cases.}
\label{fig:histogramsall}
\end{figure}

\begin{figure}
\centering
\includegraphics[width=\linewidth]{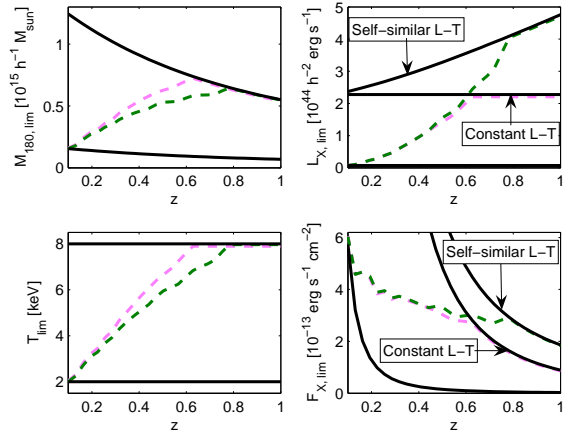}
\caption{Expected observational limits for the $^{500}$XCS (defined as ${\rm P}({\rm detection})\ge 0.5$), for the simplest case of no scatter in the cluster scaling relations and a constant (pink dashed line) or self-similar (green dashed line) $L$--$T$ relation. Solid lines correspond to the hard temperature cut $2\,{\rm keV} \le T \le 8\,{\rm keV}$. }
\label{fig:xlimits}
\end{figure}

\begin{figure}
\centering
\includegraphics[width=\linewidth]{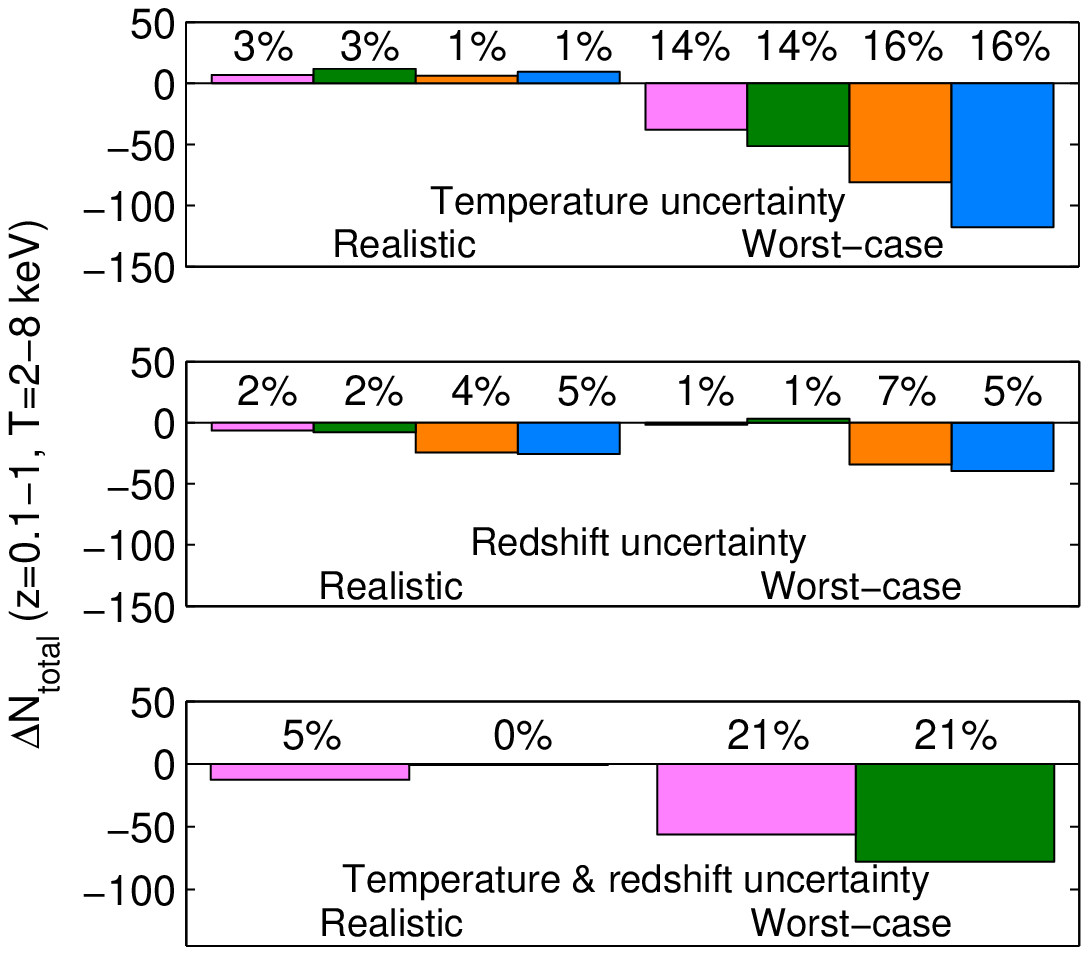}
\caption{Changes in total number of clusters due to our different error assumptions, compared to no-errors
distributions in Fig.~\ref{fig:histograms}.}
\label{fig:2derrdistcfcombo}
\end{figure}

\section{Results}
\label{Results}

\subsection{Labelling}

Generally, the colouring scheme in Table~\ref{tab:scalrel} will be
used to indicate the fiducial cluster scaling relation model; see
Sects.~\ref{tempmass} and \ref{lumtemp} as well as the Appendix.  X-ray temperature and redshift errors will be indicated as `realistic'
or `worst-case' according to Table~\ref{tab:errscheme}; see Sects.~\ref{Photozs} and \ref{Terror} as well as the Appendix.

\subsection{Expected cluster distributions}
\subsubsection{Without measurement errors}
\label{Distribnoerr}

The expected $2$D $(T, z)$ distributions of clusters for our four standard models are
shown in Fig.~\ref{fig:2dhistsnosel} (underlying distributions),
Fig.~\ref{fig:2dhists} (expected detections) and
Fig.~\ref{fig:2dhistscf} (detection efficiency), where the selection function has been used to go from Fig.~\ref{fig:2dhistsnosel} to Fig.~\ref{fig:2dhists}.  The expected
redshift distributions and total cluster number counts are shown similarly in
Fig.~\ref{fig:histogramsall}.  Note that as the $L$--$T$ relation changes, so does the expected number of detected clusters, since we are more likely to detect a cluster the more luminous it is (and for a given temperature, the cluster luminosity increases with redshift for self-similar $L$--$T$ evolution). The underlying distribution however is of course not dependent on the $L$--$T$ relation.
We find that
$^{500}$XCS can be expected to find somewhere in the range of $250$--$700$
clusters for its projected area of $500$ deg$^2$ and $0.1\le z \le 1.0$, $2\,{\rm keV} \le T \le 8\,{\rm keV}$. This corresponds to around $20$ per cent of the $1500$--$3300$ total number of clusters we would expect to detect with no photon count cut-off (effectively a $\sim 50$-photon cut-off). This full set of XCS clusters will constitute a significant sample (relative to previous studies), representing around a quarter to a third of the actual $7000$--$10000$ clusters present in the observed fields.  Going to higher redshifts, we roughly estimate that a minimum of $250$ clusters will be found at $z>1$, of which at least $10$ should have $>500$ photons.

\modsec{
The XCS DR1 currently covers an area of $132$ deg$^2$, for which $125$ clusters/groups with measured redshifts and $>500$ photons have been identified from $164$ candidate extended sources (with $>500$ photon counts). No temperature, redshift or other cuts have been applied to this set. In the current redshift sample of $125$ clusters with more than $500$ photons, approximately $40$ per cent have temperatures below $2$ keV and are therefore classified as groups. We therefore expect
the final number of genuine $T>2$~keV clusters we detect in this area to be in the range $75$--$100$, depending on the fraction of `extended' sources detected by the survey that turn out not to be clusters (blended point sources, etc.). We have here assumed that selection effects in the redshift follow-up do not significantly bias this number.  For our fiducial cosmology, we find that the corresponding expected number of clusters is $80$--$235$, the range corresponding to the lower and upper limits from our set of scaling-relation assumptions. These two ranges are clearly consistent. The lower predicted limit corresponds to no $L$--$T$ evolution, and hence -- all else equal -- the observational result might indicate a near-constant $L$--$T$ relation.}

An overview of the expected observational limits of the $^{500}$XCS for mass, X-ray temperature, X-ray luminosity and X-ray flux (in the $[0.1,2.4]$ keV band), is given in Fig.~\ref{fig:xlimits}. We have there defined the detection limit, through the selection function, as ${\rm P}({\rm detection})\ge 0.5$. These limits are thus the values above which we expect to detect, with a photon count of $500$ or above, at least half of the clusters.  It is worth noting that the change in detection probability is slow as a function of X-ray temperature, and hence the concept of e.g. a single flux limit (which would correspond to a sharp transition between one and zero in the probability) is not suitable for defining the XCS sample. The underlying reason for this is that the {\it XMM} archive images occupy a range of different exposure times, hence individual flux limits. Caution is therefore advised when comparing Fig.~\ref{fig:xlimits} to similar plots based on a single flux or mass limit. For comparison, using ${\rm P}({\rm detection})\ge 0.05$ to define the detection limit leads to a flux limit of $\sim 5 \times 10^{-14}$~erg~s$^{-1}$~cm$^{-2}$, considerably lower than that shown in Fig.~\ref{fig:xlimits}.

\subsubsection{With measurement errors}
\label{Distriberr}
Introducing measurement errors for redshift and X-ray temperature
will introduce scattering of clusters across the redshift and
temperature cut-offs. As the cluster distribution is not symmetric
with respect to these cut-offs, there may be a net increase/decrease in
the expected number of clusters as a result (a type of
Malmquist bias). Furthermore, the measurement error distributions may
also be asymmetric, as is our temperature error distribution. Note that the relevant `underlying' cluster distributions for these purposes are the expected detections, shown in Fig.~\ref{fig:2dhists}.

The change in the expected total number of clusters as
a result of different measurement error assumptions are shown in
Fig.~\ref{fig:2derrdistcfcombo}.
We find that the effect of measurement errors on the number count is significantly less than the effect of intrinsic scaling-relation scatter (cf.~Fig.~\ref{fig:histograms}). This is not surprising since the scaling-relation scatter is based on the true underlying cluster distribution in Fig.~\ref{fig:2dhistsnosel}, a much steeper function than the expected detections in Fig.~\ref{fig:2dhists}.

We also see that only in the case of worst-case temperature errors is the
Malmquist bias significant, and as we shall see later only in this
case do the measurement errors give a significant bias in cosmological
constraints, if unaccounted for.  For realistic temperature errors, a net increase in clusters is seen, as the skewness of the temperature error distribution toward low temperatures (Fig.~\ref{fig:terrdists}) is compensated by the somewhat larger number of low-temperature clusters scattering up in temperature at the low-temperature end. For worst-case temperature errors the temperature is very poorly constrained, and this compensatory effect is not sufficient to counteract the net decrease in number of clusters. Redshift errors tend to cause a loss of clusters at the low-redshift end, as the smaller cosmic volume at lower redshifts means more clusters scatter down in redshift than scatter up. However, the redshift errors also affect the temperature determination, so that low-temperature clusters scattering up could give a net increase. For realistic redshift errors the size of this induced error in temperature is $5$ per cent, which is too small to have a significant impact. For worst-case redshift errors, we see that for the case with no scaling-relation scatter, the induced temperature error of $10$ per cent reduces the loss of clusters compared to that for realistic redshift errors. For the case with scaling-relation scatter, this effect is not significant, presumably due to the much sharper drop in cluster numbers at low redshifts seen for these models (Fig.~\ref{fig:histograms}), leading to the direct redshift error dominating.

\begin{figure*}
\centering
\subfloat[Known scaling relations, no scatter]{\label{fig:contours}
\includegraphics[width=9.3cm]{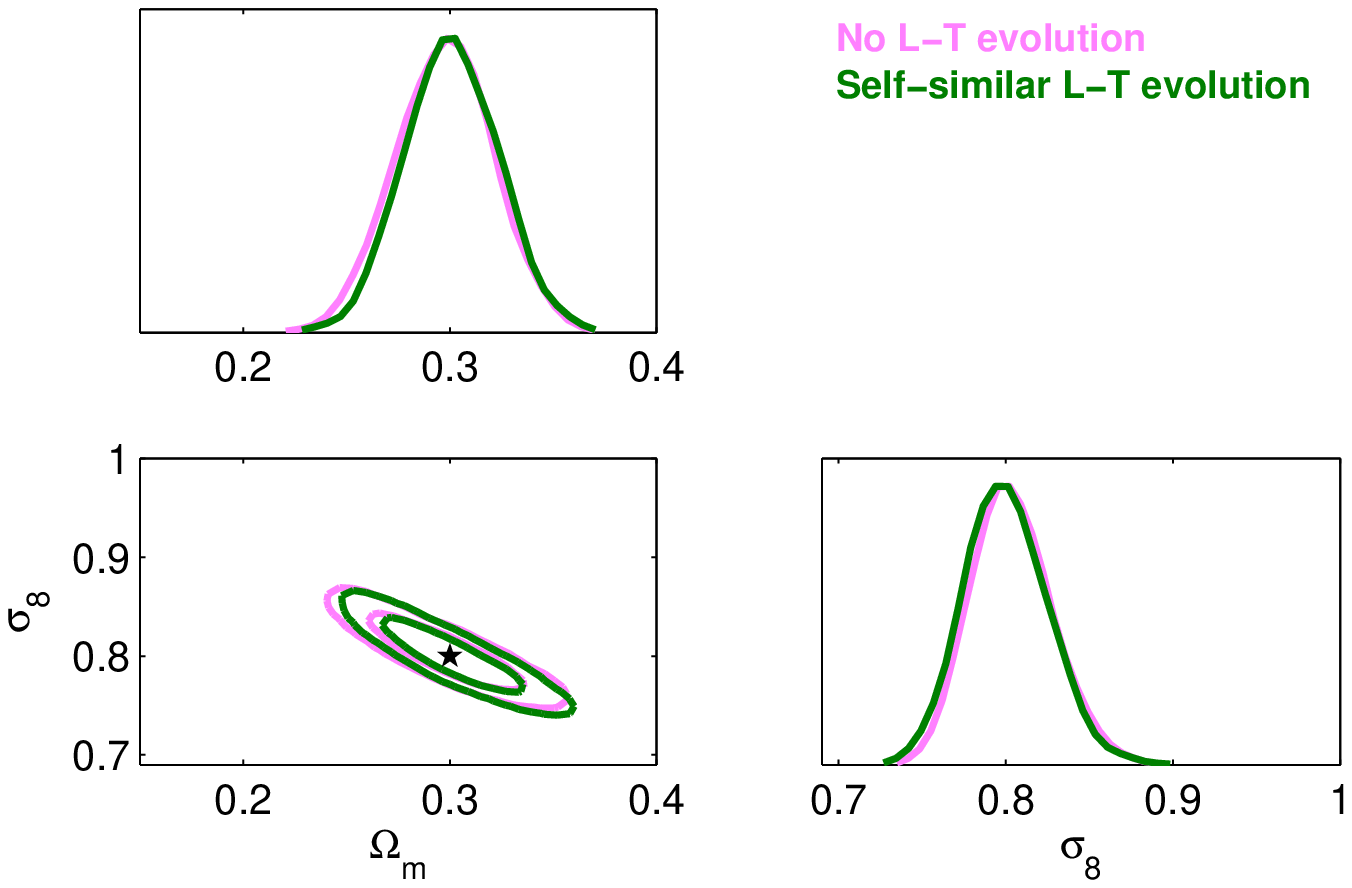}}
\subfloat[Self-calibration of $L$--$T$ relation, with
scatter]{\label{fig:scsmlcon}\includegraphics[width=9.3cm]{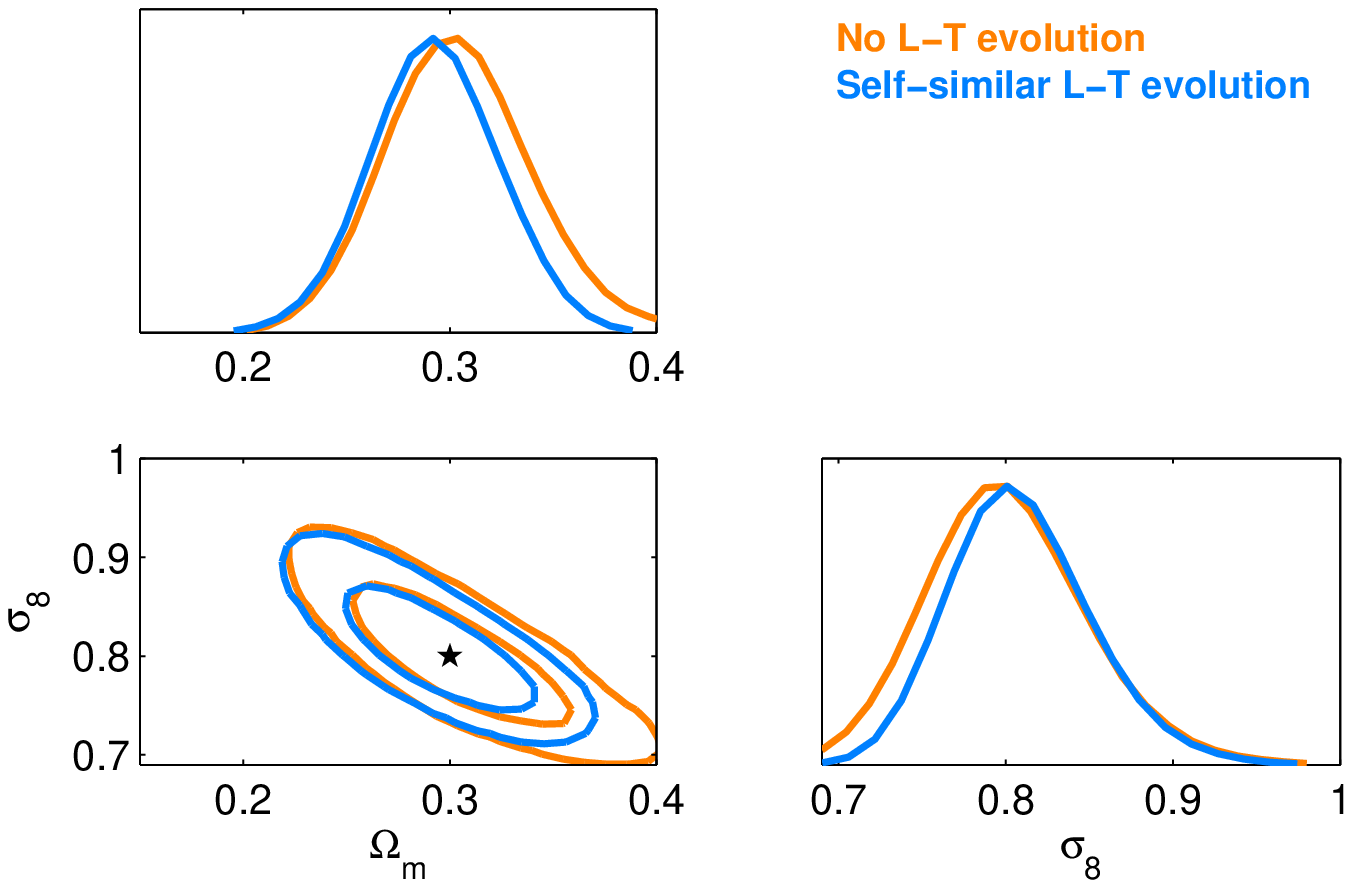}}
\caption{Expected $68$ and $95$ per cent parameter constraints for $^{500}$XCS, without measurement errors.  Stars denote the fiducial model assumed.}
\end{figure*}

The fractional change in the number of clusters is very similar
for the case with scaling-relation scatter as
without such scatter. Hence, for the case with scatter, the statistical effect will tend to be larger since the difference to the $N_{\rm ideal}$ clusters with no measurement errors relative to the Poisson error bars,
\begin{equation}
\frac{\delta N_{\rm ideal}}{\sigma((1+\delta)N_{\rm ideal})} = \frac{\delta}{\sqrt{(1+\delta)}}\sqrt{N_{\rm ideal}}\,,
\end{equation}
grows with the number of clusters (and scatter increases the number). Here, $\delta$ is the fractional change in the number of clusters. Based on this, we estimate that for all the models we consider, an upper limit on the fractional change in cluster count for a less than $1\sigma$ ($2\sigma$) bias in constraints is around $4$ ($8$) per cent, which compares favourably with the results for realistic errors in Fig.~\ref{fig:2derrdistcfcombo}. (This comparison could be made more rigorous using the Kolmogorov--Smirnov test as in \citet{HMH}, but this treatment is sufficient for our purposes.)  Due to computational limitations we have not calculated the change in number count for the case with scaling-relation scatter and both types of measurement errors, but based on the results obtained would expect them to be very similar (in fractional terms) to the results for the no-scatter case.

\subsection{Constraints: without measurement errors}
\subsubsection{Known scaling relations, no scatter}

For both choices of $L$--$T$ relation (constant and self-similar), the expected constraints are shown in Fig.~\ref{fig:contours}. We
expect $^{500}$XCS to measure $\Omega_{\rm m} = 0.3 \pm 0.02$,
$\sigma_8 = 0.8 \pm 0.02$ in each case.  The $\sigma_8$--$\Omega_{\rm m}$ degeneracy
differs somewhat between the two $L$--$T$ cases, for a constant
$L$--$T$ approximately given by
\begin{equation}
\sigma_8 = 0.8\left(\frac{\Omega_{\rm m}}{0.3}\right)^{-0.36} \,,
\end{equation}
and for a self-similar $L$--$T$ by
\begin{equation}
\sigma_8 = 0.8\left(\frac{\Omega_{\rm m}}{0.3}\right)^{-0.40}\,.
\end{equation}
These degeneracies are somewhat different from previous studies,
e.g. $\sigma_8 \propto \Omega_{\rm m}^{-0.47}$ in \citet{VL99}. That
study however used only the total number of clusters above a certain
temperature threshold to arrive at constraints. The orientation also depends on redshift depth \citep{LSW}. These constraints are better than what has been forecast for {\it XMM}--LSS \citep{RVP}, but the comparison is not entirely like-for-like as they employ the Press--Schechter mass function and assume a rather different fiducial $\sigma_8$ and $\Gamma$. The constraints are also fairly competitive with what can be expected from other surveys using e.g. the South Pole Telescope (SPT), {\it Planck} or {\it DUET} \citep{MJ04,Geisbuesch_planck}, but in making this comparison one should note that we employ much more restrictive priors; the set of free parameters is not exactly the same.

\begin{figure}
\includegraphics[width=\linewidth]{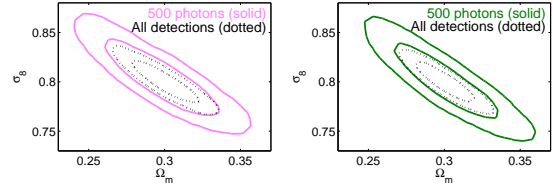}
\caption{Comparison of $^{500}$XCS to the case where all detections
are used. A constant (left) or self-similar (right) $L$--$T$ relation,
and no $L$--$T$ or $M$--$T$ scatter was assumed. Contours correspond to $68$ and $95$ per cent confidence levels.}
\label{fig:contoursxcombo}
\end{figure}

\begin{table*}
\centering
\begin{tabular}{|c|c|c|c|}
  \hline
  $L$--$T$ evolution & & Known scaling relations, no scatter &
  Self-calibration of $L$--$T$, with scatter   \\
  \hline
  Constant & {\begin{tabular}{c}
              $\Omega_{\rm m}$ \\
              $\sigma_8$ \\
              $\sigma_{\log L_{\rm X}}$ \\
              $\alpha$ \\
              $\beta$ \\
              $\gamma_z$
            \end{tabular}
  } & {\begin{tabular}{c}

                  $0.30 \pm 0.02$ \\
                  $0.80 \pm 0.02$ \\
                  -- \\
                  -- \\
                  -- \\
                  --
              \end{tabular}} &
{\begin{tabular}{c}

                $0.30 \pm 0.03$\\
                $0.80 \pm 0.05$ \\
                  $[0.2,0.4]$ \\
                  $-1.91 \pm 0.12$ \\
                  $2.50 \pm 0.33$ \\
                  $[-1,1.5]$
              \end{tabular}}
                 \\
                 \hline
  Self-similar & {\begin{tabular}{c}

              $\Omega_{\rm m}$ \\
              $\sigma_8$ \\
              $\sigma_{\log L_{\rm X}}$ \\
              $\alpha$ \\
              $\beta$ \\
              $\gamma_z$
            \end{tabular}
  }
  &
{\begin{tabular}{c}

                $0.30 \pm 0.02$ \\
                $0.80 \pm 0.02$ \\
                  -- \\
                  -- \\
                  -- \\
                --
              \end{tabular}}
  & {\begin{tabular}{c}

                $0.30 \pm 0.03$ \\
                $0.80 \pm 0.04$ \\
                 $[0.2,0.4]$ \\
                 $-1.92 \pm 0.12$ \\
                  $2.55 \pm 0.31$ \\
                  $[-1,1.5]$
              \end{tabular}}
              \\
  \hline
\end{tabular}
\caption{Expected $1\sigma$ parameter constraints for $^{500}$XCS when marginalized over all other parameters, without measurement errors.}
\label{tab:noerrs}
\end{table*}

The constraints in Fig.~\ref{fig:contours} are for a photon-count
threshold of 500. Lowering the photon-count threshold so that more
clusters are included in the sample should clearly affect
the size of constraints. We find that using all detections
(corresponding to an effective photon-count threshold of typically $\sim 50$ photons) improves 1D
constraints by about $40$ per cent (Fig.~\ref{fig:contoursxcombo}). This corresponds to an increase in the number of clusters used of around $1200$--$1700$ ($400$--$500$ per cent).
For clusters with few photon counts the temperature errors will become very large
\citep{LVRM}. Contamination from e.g. galaxy groups will also rise sharply with decreasing
photon-count threshold, partly because clusters with low photon count will tend to have a low temperature. Hence, these estimates provide only upper
limits on the possible constraint improvement.
Taking error and contamination effects into account, it is likely that there would be only a weak improvement by including those XCS clusters expected to have a photon count below $500$. However, follow-up observations with e.g. {\it XMM} or {\it IXO} (formerly {\it XEUS}) could improve the photon statistics of those clusters enough to make their inclusion in the analysis worthwhile. We discuss this in more detail in Sect.~\ref{Conclusions}.

\subsubsection{Self-calibration of $L$--$T$ relation, with scatter}
Self-calibration is the process by which e.g. the $L$--$T$ relation
can be constrained jointly with cosmological parameters using only the
($T,z$) cluster number counts \citep{Hu,LHI,MJ04,LH}.

\begin{figure*}
\centering
\subfloat[Constant $L$--$T$ relation]{\label{fig:scnoevolcon}\includegraphics[width=14cm]{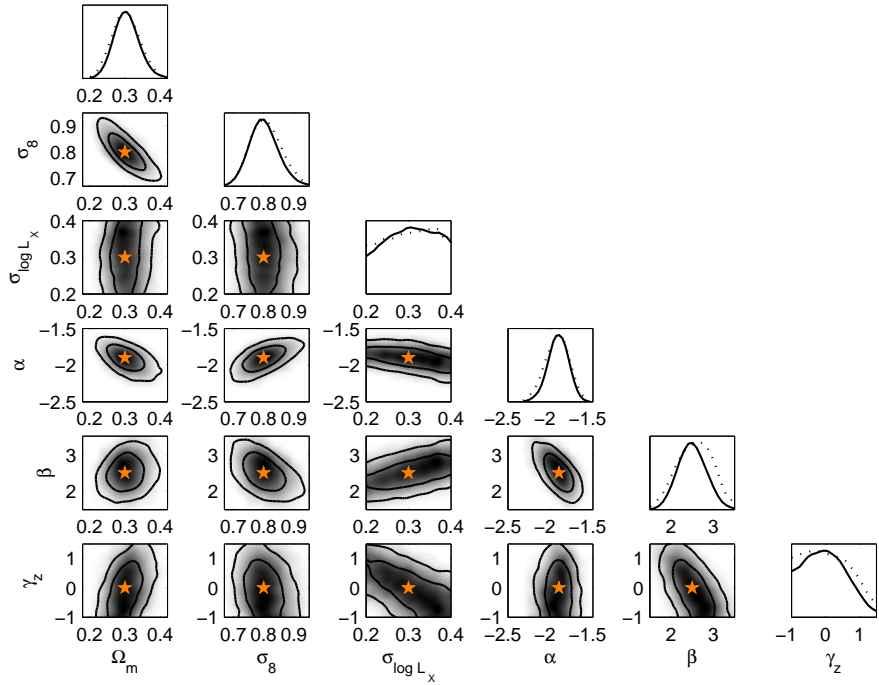}} \\
\subfloat[Self-similar $L$--$T$ relation. (As can be surmised from some of the $1$D distributions, the marginalized and mean likelihoods approach each other very slowly in the MCMC due to the prior cutting the distribution, however the statistical properties of the distribution have converged appropriately.)]{\label{fig:scselfsimcon}\includegraphics[width=14cm]{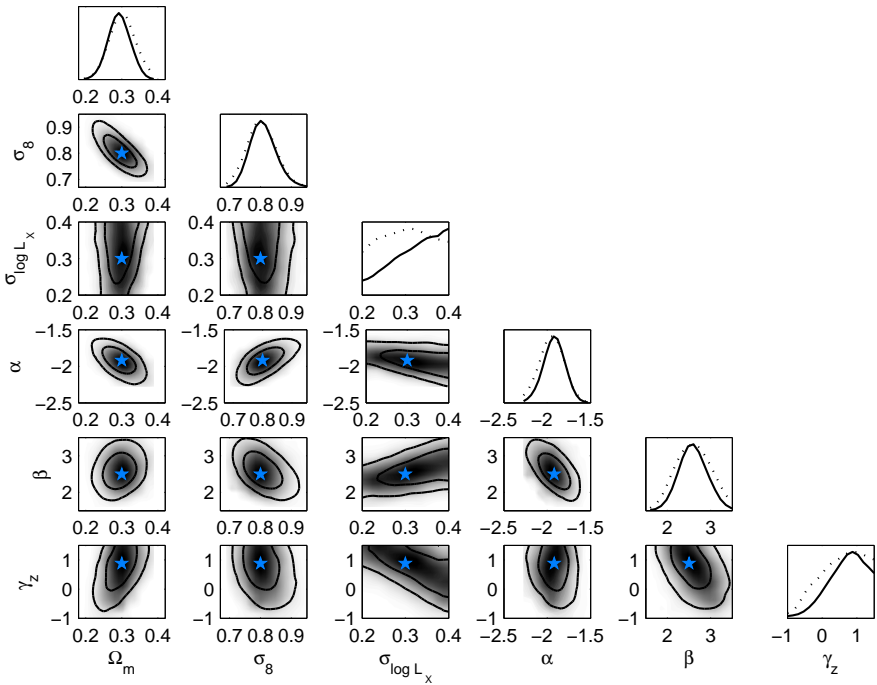}}
\caption{Expected $68$ and $95$ per cent parameter constraints for
$^{500}$XCS, with scaling-relation scatter
and no measurement errors, and fitting jointly with $L$--$T$ relation
(self-calibration) for which reasonable priors on scatter and redshift
evolution have been adopted. Solid lines correspond to marginalized
likelihood, dotted lines and shading to mean likelihood.  Stars denote the fiducial model assumed.}
\label{fig:selfcalall}
\end{figure*}

\begin{figure*}
\centering
\subfloat[Accounted-for single measurement errors]{\label{fig:terrsfitted}\includegraphics[width=9.3cm]{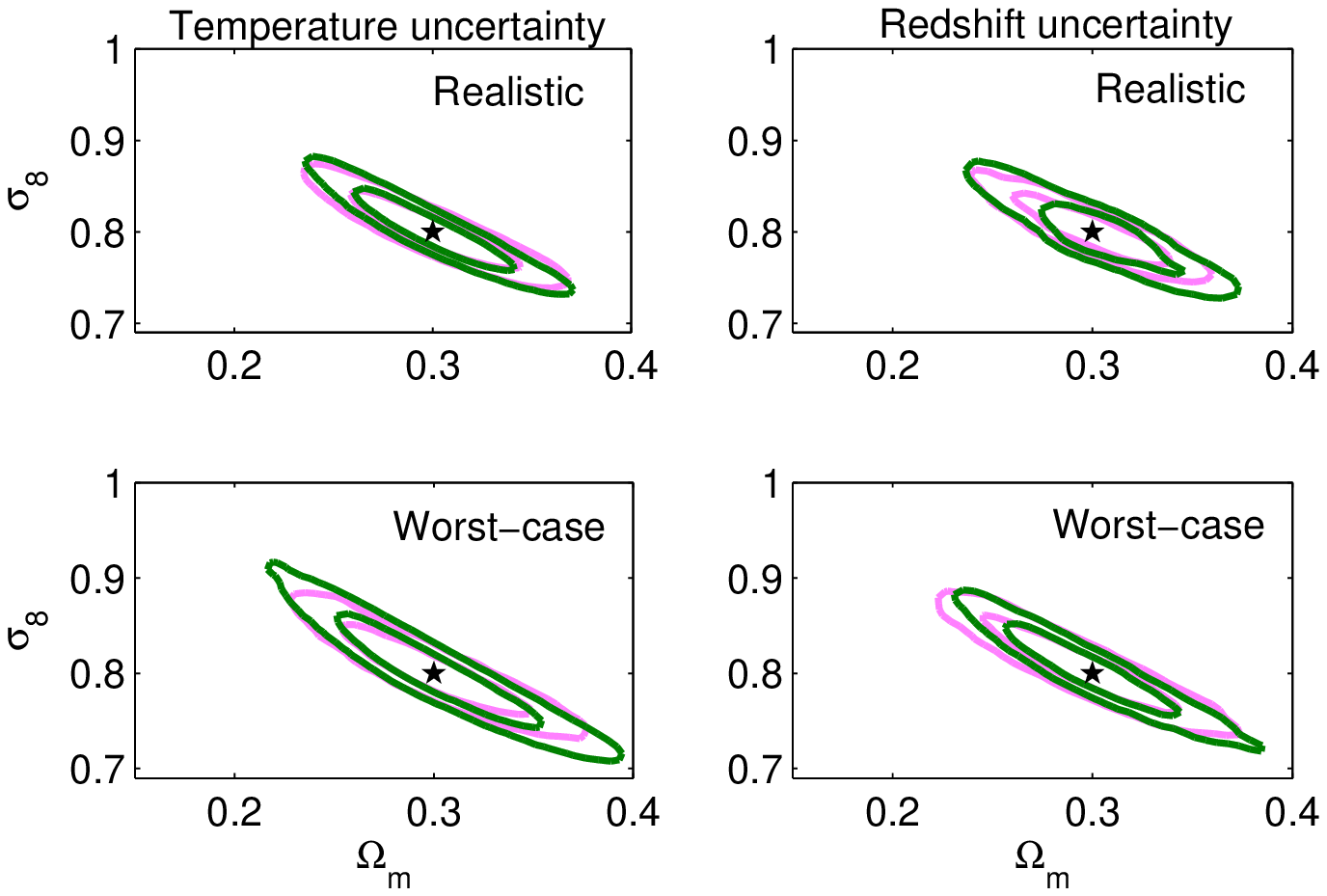}}
\subfloat[Unaccounted-for
single measurement
errors]{\label{fig:terrsnotfitted}
\includegraphics[width=9.3cm]{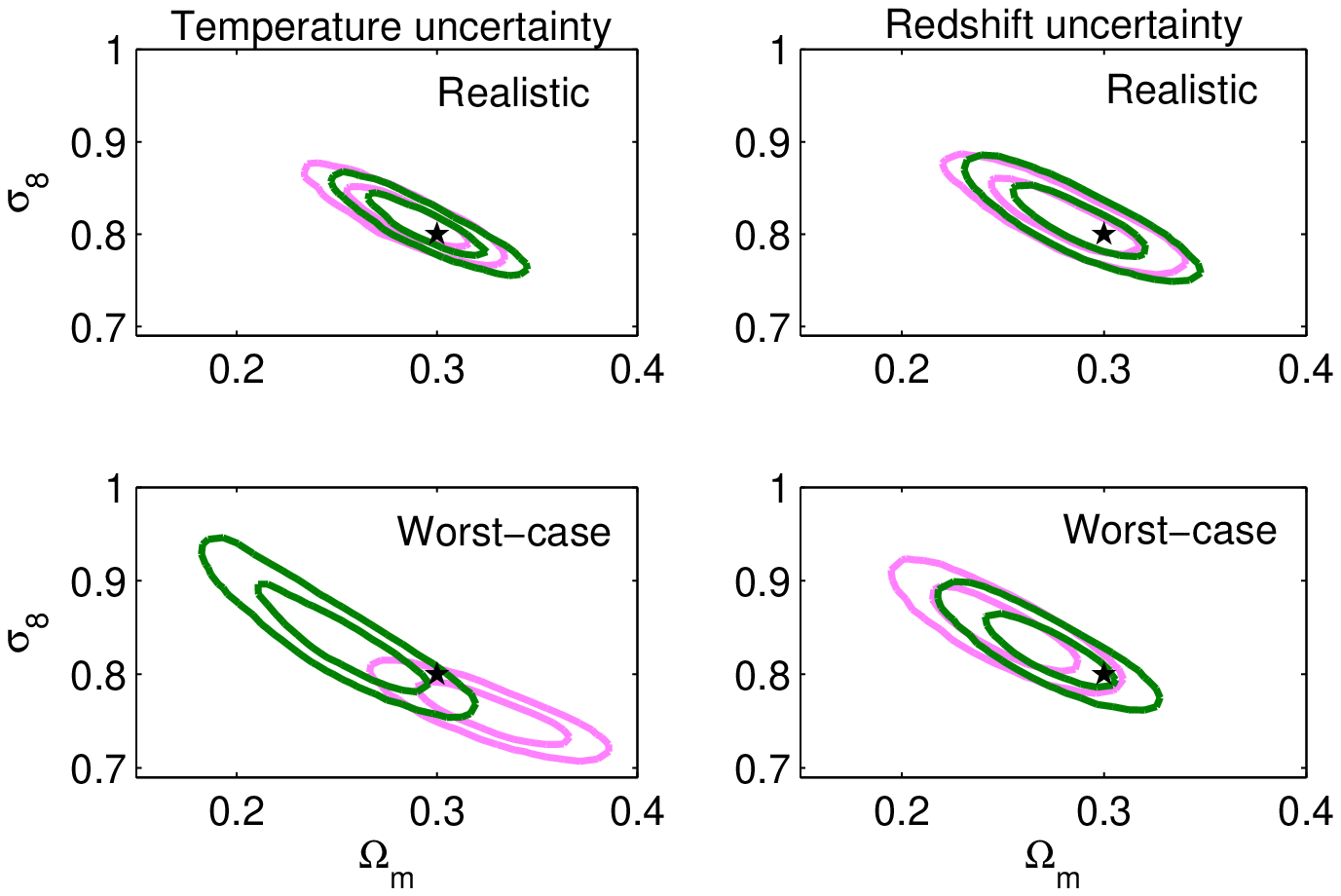}}
\\ \subfloat[Unaccounted-for
combined measurement errors]{\label{fig:terrsdblnotfitted}
\includegraphics[width=9.3cm]{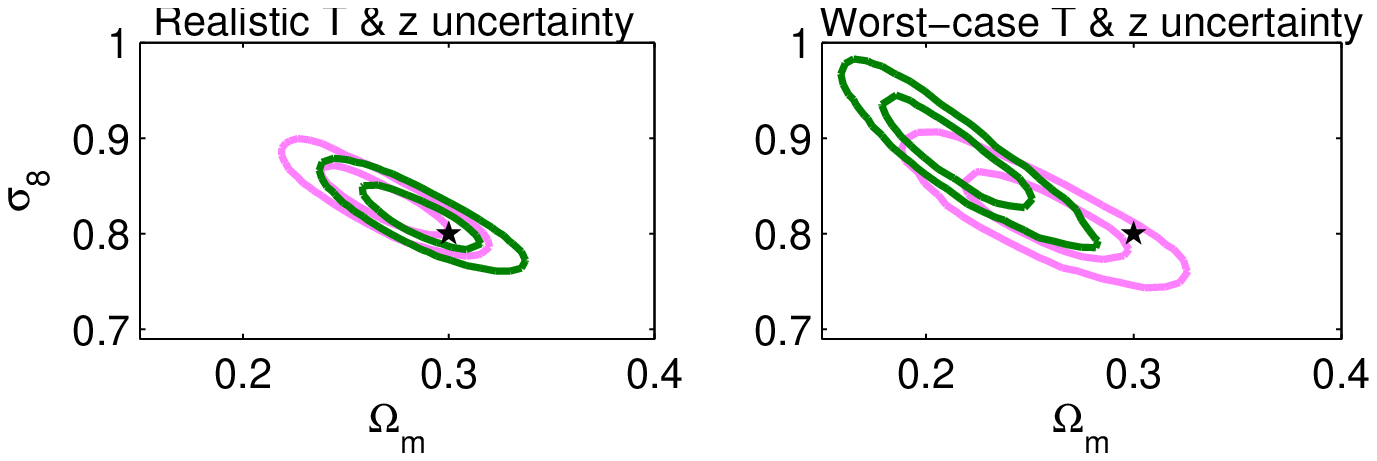}}
\caption{Expected $68$ and $95$ per cent parameter constraints for $^{500}$XCS, for
known scaling relations, no scatter, and with measurement errors.  Stars denote the fiducial model assumed.}
\end{figure*}

\begin{table*}
\centering
\begin{tabular}{|c|c|c|c|c|c|}
  \hline
  $L$--$T$ evolution & & Realistic $T$ errors & Worst-case $T$ errors
  & Realistic $z$ errors & Worst-case $z$ errors \\
  \hline
  Constant & {\begin{tabular}{|c|}
              $\Omega_{\rm m}$ \\
              $\sigma_8$ \\

            \end{tabular}
  } & {\begin{tabular}{|c|}

                  $0.30 \pm 0.03$ \\
                  $0.80 \pm 0.03$ \\

              \end{tabular}} &
{\begin{tabular}{|c|}

                $0.30 \pm 0.03$\\
                $0.80 \pm 0.03$ \\

              \end{tabular}}
               &
{\begin{tabular}{|c|}

                $0.30 \pm 0.02$ \\
                $0.80 \pm 0.02$ \\

              \end{tabular}}
                & {\begin{tabular}{|c|}

                $0.30 \pm 0.03$\\
                $0.80 \pm 0.03$ \\

              \end{tabular}}
                 \\
  Self-similar & {\begin{tabular}{|c|}

              $\Omega_{\rm m}$ \\
              $\sigma_8$ \\

            \end{tabular}
  }
  &
{\begin{tabular}{|c|}

                $0.30 \pm 0.03$ \\
                $0.80 \pm 0.03$ \\

              \end{tabular}}
  & {\begin{tabular}{|c|}

                $0.30 \pm 0.03$ \\
                $0.80 \pm 0.04$ \\

              \end{tabular}}
              & {\begin{tabular}{|c|}

                $0.30 \pm 0.03$\\
                $0.80 \pm 0.03$ \\

              \end{tabular}}
               &
              {\begin{tabular}{|c|}

                $0.30 \pm 0.03$ \\
                $0.80 \pm 0.03$ \\

              \end{tabular}} \\
  \hline
\end{tabular}
\caption{Expected $1\sigma$ parameter constraints for $^{500}$XCS when marginalized over the other parameter, for
known scaling relations, no scatter, and with accounted-for measurement errors.}
\label{tab:noscerrs}
\end{table*}

\begin{table*}
\centering
\begin{tabular}{|c|c|c|c|c|c|c|c|}
  \hline
  $L$--$T$ evolution\hspace{-0.5cm} & \hspace{-0.5cm} &
  {\begin{tabular}{c}
     Realistic \\
    $T$ errors \end{tabular}}
    \hspace{-0.5cm}
     &
{\begin{tabular}{c}
     Worst-case \\
   $T$ errors \end{tabular}}
   \hspace{-0.5cm}
     &
{\begin{tabular}{c}     Realistic \\
   $z$ errors \end{tabular}}
   \hspace{-0.5cm}
     &
{\begin{tabular}{c}
    Worst-case \\
     $z$ errors \end{tabular}}
     \hspace{-0.5cm}
  & {\begin{tabular}{c}
    Realistic \\
    $T$ \& $z$ errors \end{tabular}}
    \hspace{-0.5cm}
  & {\begin{tabular}{c}
    Worst-case \\
   $T$ \& $z$ errors \end{tabular}}
   \hspace{-0.5cm}
    \\
  \hline
  Constant \hspace{-0.5cm} & {\begin{tabular}{c}
              $\Omega_{\rm m}$ \\
              $\sigma_8$
            \end{tabular}
  }\hspace{-0.5cm} & {\begin{tabular}{c}
                  $0.28 \pm 0.02$ \\
                  $0.82 \pm 0.02$
              \end{tabular}}
                   \hspace{-0.5cm}
               &
{\begin{tabular}{c}
                $0.33 \pm 0.02$\\
                $0.76 \pm 0.02$
              \end{tabular}}
                   \hspace{-0.5cm}
               &
{\begin{tabular}{c}
                $0.28 \pm 0.02$ \\
                $0.82 \pm 0.03$
              \end{tabular}}
                   \hspace{-0.5cm}
                & {\begin{tabular}{c}
                $0.25 \pm 0.02$\\
                $0.85 \pm 0.03$
              \end{tabular}}
                   \hspace{-0.5cm}
& {\begin{tabular}{c}
                  $0.27 \pm 0.02$ \\
                  $0.83 \pm 0.02$
              \end{tabular}} \hspace{-0.5cm}
              &
{\begin{tabular}{c}
                $0.26 \pm 0.03$\\
                $0.82 \pm 0.03$
              \end{tabular}}
                   \hspace{-0.5cm}
                 \\
  Self-similar\hspace{-0.5cm} & {\begin{tabular}{c}
              $\Omega_{\rm m}$ \\
              $\sigma_8$
            \end{tabular}
  }\hspace{-0.5cm}
  &
{\begin{tabular}{c}
                $0.29 \pm 0.02$ \\
                $0.81 \pm 0.02$
              \end{tabular}}
                   \hspace{-0.5cm}
  & {\begin{tabular}{c}
                $0.25 \pm 0.03$ \\
                $0.84 \pm 0.04$
              \end{tabular}}
                   \hspace{-0.5cm}
              & {\begin{tabular}{c}
                $0.29 \pm 0.02$\\
                $0.81 \pm 0.03$
              \end{tabular}}
                   \hspace{-0.5cm}
               &
              {\begin{tabular}{c}
                $0.27 \pm 0.02$ \\
                $0.83 \pm 0.03$
              \end{tabular}}
                \hspace{-0.5cm}
            &
{\begin{tabular}{c}
                $0.29 \pm 0.02$ \\
                $0.82 \pm 0.02$
              \end{tabular}}
              \hspace{-0.5cm}
  & {\begin{tabular}{c}
                $0.22 \pm 0.02$ \\
                $0.88 \pm 0.04$
              \end{tabular}}
              \hspace{-0.5cm}
              \\
  \hline
\end{tabular}
\caption{Expected $1\sigma$ parameter constraints for $^{500}$XCS when marginalized over the other parameter, for known scaling relations, no scatter, and with unaccounted-for measurement errors.}
\label{tab:noscerrnoerrs}
\end{table*}

We find that jointly fitting for the cosmological parameters and the
$L$--$T$ relation, $^{500}$XCS will measure $\Omega_{\rm m} = 0.3 \pm 0.03,
\sigma_8 = 0.8 \pm 0.05$ under our assumptions. The marginalized $\Omega_{\rm m}$--$\sigma_8$ likelihood distributions are shown in Fig.~\ref{fig:scsmlcon}, and the full set of likelihood distributions in Fig.~\ref{fig:selfcalall} (note that Fig.~\ref{fig:scsmlcon} is just the top triangles of these plots). The $1$D parameter constraints are listed in Table~\ref{tab:noerrs}. The constraints for the case of self-similar $L$--$T$ evolution appear narrower than for a constant $L$--$T$. This is due to the redshift-evolution prior, explained below, significantly cutting the distribution. We thus believe the constant $L$--$T$ case to be most representative of the constraints we can expect.
As expected, the constraints on $\Omega_{\rm m}$ and $\sigma_8$ degrade when marginalizing over the four $L$--$T$ parameters (compared to Fig.~\ref{fig:contours}), but still remain relatively small. In comparison to the South Pole Telescope, {\it Planck}, and {\it DUET} \citep{MJ04,Geisbuesch_planck}, our constraints are still competitive (we lack comparable results for {\it XMM}--LSS, but expect to do better given our larger survey area and depth). However, if we were to consider self-calibration of the $M$--$T$ relation as well (rather than using an external description, as described in Sect.~\ref{tempmass}), those surveys would have more power than the XCS (using only archival {\it XMM} data) through the use of the cluster power spectrum \citep{MJ04,LHI}. In fact, we do not expect XCS to have any significant constraining power if the $M$--$T$ relation is self-calibrated as well. We show examples of the effects of $M$--$T$ systematics in Sect.~\ref{systbiases}.
It has been shown \citep[e.g.][]{MJ04} that small follow-up samples can dramatically improve the situation. Therefore, weak-lensing/SZ follow-up and/or a contiguous e.g. {\it XMM} survey would be highly advantageous \citep[see also][]{Berge,XMMLSSfuture}.
Comparing to Fig.~\ref{fig:contours}, although we lose constraining power due to an increase in the number of parameters, since we are including scaling-relation scatter the number of clusters increases significantly which mitigates the degradation.  Note that, as shown in Table~\ref{tab:scalrel}, we fit the data to a power-law $L$--$T$ relation $\sim (1+z)^{\gamma_z}$. Although the functional form for a self-similar $L$--$T$ used to generate data is different in principle, we have checked that a power law can approximate its redshift evolution very well.

Using $(T, z)$ number-count self-calibration, based only on archival {\it XMM} data (Fig.~\ref{fig:selfcalall}), we can constrain the $L$--$T$ normalization $\alpha$ to $\pm 0.12$ (or $\pm 6$ per cent) and the $L$--$T$ slope $\beta$ to $\sim \pm 0.3$ (or $\pm 13$ per cent). The self-calibration procedure is not able to jointly constrain the scatter $\sigma_{\log L_{\rm X}}$ and redshift evolution $\gamma_z$ significantly.
We have therefore imposed flat
priors on these parameters, $0.2 \le \sigma_{\log L_{\rm X}} \le 0.4$
and $-1 \le \gamma_z \le 1.5$ to limit the distribution within
reasonable bounds of a size reflecting the minimum accuracy to which we would hope to measure these parameters from our direct $L$--$T$ data, i.e. also taking into account the measured X-ray flux (see also Table~\ref{tab:scalrel}).

Thus, the self-calibration power to constrain the $L$--$T$ relation is present in the data, but as can be seen in Fig.~\ref{fig:selfcalall} there are strong degeneracies between parameters.
The main degeneracy is that between $\gamma_z$ and
$\sigma_{\log L_{\rm X}}$; increasing $\sigma_{\log L_{\rm X}}$ can
easily be offset by reducing $\gamma_z$, which also is easy to understand physically as they both effectively scale the cluster luminosities up or down, and corresponds to the observation by several authors \citep[e.g.][]{Branchesietal,MauYX,Nord,XMMLSSfirst} that $L$--$T$ scatter can mimic $L$--$T$ evolution (also discussed in Sect.~\ref{lumtempevol}).
The redshift evolution $\gamma_z$ is also degenerate with the
$L$--$T$ slope $\beta$, which is thus itself degenerate with $\sigma_{\log
L_{\rm X}}$. Interestingly though, the cosmological parameters show little degeneracy with $\sigma_{\log L_{\rm X}}$.
It is the result of these degeneracies that all four $L$--$T$ parameters cannot be
jointly constrained. Bayesian Complexity \citep*{Kunzcomplexity}
suggests that at most five parameters (including $\Omega_{\rm m}$ and $\sigma_8$) can be fully constrained, which is also what we find in practice. As one might expect, we will therefore have to rely on our direct $L$--$T$ measurement to constrain the $L$--$T$ scatter and evolution (as proposed by \citealt*{VHS}; \citealt{Hu,BW,Wang04,LH}).

The fact that our relatively generous priors on the $L$--$T$ scatter and evolution still restricts the distribution, affecting the size of cosmological constraints, also serves to illustrate a slightly different point of view: turning the problem around, and using complementary cosmological data to constrain e.g. $\Omega_{\rm m}$ and $\sigma_8$, thereby possibly also improving constraints on astrophysical parameters \citep[as noted by e.g.][]{LSW,HK,Hu}.

\subsection{Constraints: with measurement errors}
\subsubsection{Known scaling relations, no scatter}

The effect on derived cosmological constraints from measurement errors
in X-ray temperature and redshift is small.
Taking into account knowledge of the error distributions in the data analysis,
we find that the size of uncertainties increases
somewhat compared to the no-errors case (see
Fig.~\ref{fig:terrsfitted} and Table~\ref{tab:noscerrs}, cf. Fig.~\ref{fig:contours} and column~1 in Table~\ref{tab:noerrs}). Interestingly, even with temperature or redshift errors of an unrealistically large magnitude, the effect on the constraints is small.  As such, we expect the broadening of constraints due to measurement errors to be a minor effect compared to the effects of possible systematic errors. These findings are in agreement with what has already been found by e.g. \citet{Hut04,Hut06,LimaHuphotoz}.

\begin{figure}
\centering
\includegraphics[width=9cm]{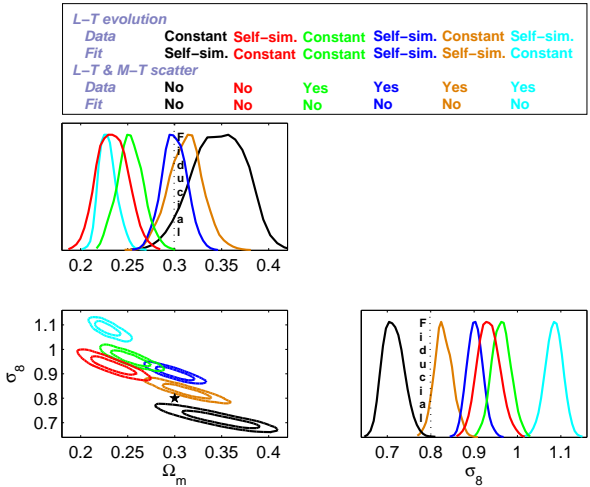}
\caption{Expected $68$ and $95$ per cent parameter constraints from the $^{500}$XCS, for various cluster scaling relation assumptions inconsistent with the fiducial model used
for generating the data. The different data and fitting assumptions are colour coordinated with the contours, and listed in the panel above the plot. The model parameters are the same as previously, and listed in Table~\ref{tab:scalrel}. The corresponding cluster distributions in redshift and temperature can be found in Fig.~\ref{fig:2dhists}.} \label{fig:biascontours}
\end{figure}

\begin{table*}
\centering
\begin{tabular}{|c|c|c|c|c|c|c|}
  \hline
  &
  {\begin{tabular}{c} 1
     \end{tabular}}
     &
{\begin{tabular}{c} 2
     \end{tabular}}
     &
{\begin{tabular}{c}  3 \end{tabular}}
     &
{\begin{tabular}{c}
    4 \end{tabular}}
  & {\begin{tabular}{c}
    5 \end{tabular}}
  & {\begin{tabular}{c}
    6 \end{tabular}}
    \\
  \hline
  {\begin{tabular}{c}
              $\Omega_{\rm m}$ \\
              $\sigma_8$
            \end{tabular}
  } & {\begin{tabular}{c}
                  $0.35 \pm 0.03$ \\
                  $0.71 \pm 0.02$
              \end{tabular}}
               &
{\begin{tabular}{c}
                $0.23 \pm 0.02$\\
                $0.93 \pm 0.02$
              \end{tabular}}
               &
{\begin{tabular}{c}
                $0.25 \pm 0.01$ \\
                $0.96 \pm 0.02$
              \end{tabular}}
                & {\begin{tabular}{c}
                $0.30 \pm 0.01$\\
                $0.90 \pm 0.02$
              \end{tabular}}
& {\begin{tabular}{c}
                  $0.31 \pm 0.02$ \\
                  $0.83 \pm 0.02$
              \end{tabular}}
              &
{\begin{tabular}{c}
                $0.23 \pm 0.01$\\
                $1.08 \pm 0.02$
              \end{tabular}}
  \\
  \hline
\end{tabular}
\caption{Expected $1\sigma$ parameter constraints for $^{500}$XCS when marginalized over the other parameter, for systematic errors in the scaling-relation assumptions (Fig.~\ref{fig:biascontours}). The different scenarios are numbered according to the order in which they are listed in Fig.~\ref{fig:biascontours}.}
\label{tab:ltmtbias}
\end{table*}

The effect of ignoring
temperature and redshift errors in the fitting procedure can to some
extent model one such systematic; poor knowledge of the measurement error distributions.  As can be seen in Fig.~\ref{fig:terrsnotfitted}, we find that when ignoring measurement
errors in the fitting, for all combinations of single measurement
errors (i.e. only $z$ or $T$ at a time), the difference in
cosmological constraints compared to the fiducial model
is within $2\sigma$ (and most are within $1\sigma$). For
combined $z$ and $T$ measurement errors, the same is still true for
realistic errors, but for a self-similar $L$--$T$ and worst-case
errors the bias is larger than $2\sigma$ (see
Fig.~\ref{fig:terrsdblnotfitted}). These results agree well with the expectations presented in Sect.~\ref{Distriberr}, and thus suggest that a good estimate of the bias in cosmological constraints due to Malmquist-bias effects can be obtained by comparing the net Malmquist bias to the Poisson error of the total cluster number count (at least to roughly discriminate $>2\sigma$ bias from $<2\sigma$ bias). This is not that surprising as the shape of the cluster distribution does not differ much between such models, and thus the total number count carries most of the information \citep[also noted in][]{HMH}.
The $1$D constraints corresponding to Figs.~\ref{fig:terrsnotfitted}~\&~\ref{fig:terrsdblnotfitted} are listed in
Table~\ref{tab:noscerrnoerrs}.

\subsubsection{Self-calibration of $L$--$T$ relation, with scatter}

Because of computational limitations we have not explicitly calculated
cosmological constraints for self-calibration with measurement
errors. We have however checked that when scatter is
included in the data, the effect of temperature and redshift errors on
the expected cluster distribution is very similar to the case where no
scatter is included, see Fig.~\ref{fig:2derrdistcfcombo}, and the discussion in the preceding Section and Sect.~\ref{Distriberr}. We thus
expect that the effect from measurement errors on constraints where
scatter is included, with or without self-calibration, can be expected
to be small or negligible -- both in terms of bias if the errors are
ignored, or broadening of error contours when errors are taken into
account. We therefore believe that the self-calibration results for
the case without measurement errors (Figs.~\ref{fig:scsmlcon}~\&~\ref{fig:selfcalall}, Table~\ref{tab:noerrs}) should provide a good rough
approximation of the expected self-calibration constraints with
measurement errors.  Note that this situation is bound to change once direct $L$--$T$ data is added to the procedure, as the temperature errors will then have a significant impact on the accuracy to which the evolution of the $L$--$T$ relation can be determined, hence setting the size of the constraints on $\sigma_{\log L_{\rm X}}$ and $\gamma_z$. One can therefore \emph{not} conclude that temperature errors are largely unimportant for the cosmological constraints we will ultimately produce from the data, but an upper limit on the size is set by this work \citep[see e.g.][]{VHS,Hu,BW}.

\subsection{Constraints: systematic biases}
\label{systbiases}
It is clear from the above sections that measurement errors in the guises we consider are not expected to be a major source of bias or degradation of constraints
vis-\`{a}-vis the underlying cluster distribution. However, if incorrect
assumptions as to the characteristics of the $M$--$T$ and $L$--$T$
relations are used when fitting the data, significant bias may occur,
as seen in Figs.~\ref{fig:biascontours}~\&~\ref{fig:biascontoursmt} \modsec{and Tables~\ref{tab:ltmtbias}~\&~\ref{tab:mtbias}}.

Looking first at Fig.~\ref{fig:biascontours} \modsec{(and Table~\ref{tab:ltmtbias} for the 1D marginalized constraints)}, the figure shows how both the size and best-fitting values of cosmological constraints are affected when ignoring scatter in the scaling relations, using an inappropriate $L$--$T$ relation, or both. The first case (from left in the panel above the plot) shows how using a self-similar $L$--$T$ to fit data coming from a constant $L$--$T$ leads to an overestimation of $\Omega_{\rm m}$. Comparing the second, third and fourth cases, we can see that $L$--$T$ evolution and scaling-relation scatter all have a similar effect when unaccounted for, all leading to an overestimation of $\sigma_8$ (and consequently underestimation of $\Omega_{\rm m}$). As they both have a similar effect, the self-similar evolution in the fifth case can mimic some of the unaccounted-for scatter, leading to a lesser overestimation than for the previous cases. On the other hand, the sixth and last case combines the two effects thus leading to a dramatic overestimation of $\sigma_8$. As such, this last case represents a worst-case scenario for this type of bias.

\begin{figure}
\centering
\includegraphics[width=9cm]{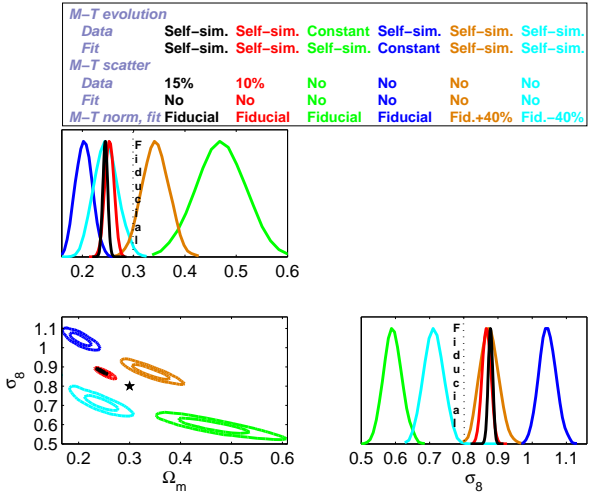}
\caption{Expected $68$ and $95$ per cent parameter constraints from the $^{500}$XCS, for various mass--temperature-relation assumptions inconsistent with the fiducial model used
for generating the data. The different data and fitting assumptions are colour coordinated with the contours, and listed in the panel above the plot. Throughout, a self-similar $L$--$T$ relation with scatter (as specified in Table~\ref{tab:scalrel}) has been assumed.} \label{fig:biascontoursmt}
\end{figure}

\begin{table*}
\centering
\begin{tabular}{|c|c|c|c|c|c|c|}
  \hline
  &
  {\begin{tabular}{c}
    1 \end{tabular}}
     &
{\begin{tabular}{c}
     2 \end{tabular}}
     &
{\begin{tabular}{c}  3 \end{tabular}}
     &
{\begin{tabular}{c}
    4 \end{tabular}}
  & {\begin{tabular}{c}
    5 \end{tabular}}
  & {\begin{tabular}{c}
    6 \end{tabular}}
    \\
  \hline
  {\begin{tabular}{c}
              $\Omega_{\rm m}$ \\
              $\sigma_8$
            \end{tabular}
  } & {\begin{tabular}{c}
                  $0.245 \pm 0.005$ \\
                  $0.878 \pm 0.007$
              \end{tabular}}
               &
{\begin{tabular}{c}
                $0.253 \pm 0.008$\\
                $0.87 \pm 0.01$
              \end{tabular}}
               &
{\begin{tabular}{c}
                $0.47 \pm 0.05$ \\
                $0.59 \pm 0.03$
              \end{tabular}}
                & {\begin{tabular}{c}
                $0.20 \pm 0.01$\\
                $1.04 \pm 0.02$
              \end{tabular}}
& {\begin{tabular}{c}
                  $0.34 \pm 0.02$ \\
                  $0.87 \pm 0.03$
              \end{tabular}}
              &
{\begin{tabular}{c}
                $0.24 \pm 0.03$\\
                $0.72 \pm 0.04$
              \end{tabular}}
              \\
  \hline
\end{tabular}
\caption{Expected $1\sigma$ parameter constraints for $^{500}$XCS when marginalized over the other parameter, for systematic errors in the $M$--$T$ relation assumptions (Fig.~\ref{fig:biascontoursmt}). The different scenarios are numbered according to the order in which they are listed in Fig.~\ref{fig:biascontoursmt}.}
\label{tab:mtbias}
\end{table*}

The other figure, Fig.~\ref{fig:biascontoursmt} \modsec{(and Table~\ref{tab:mtbias} for the 1D marginalized constraints)}, shows how constraint size and best-fitting values vary with systematic errors in the $M$--$T$ relation only. The first two cases (from left in the panel above the plot) illustrates significantly underestimating a scatter of $10$ or $15$ per cent (deviations similar to what might be expected according to \citealt{Vik06}). This leads to an overestimation of $\sigma_8$, and relatively narrow constraints, since scatter significantly increases the number of detected clusters. The largest impact seen in this figure comes from poor knowledge of the redshift evolution of the $M$--$T$ relation, seen in the second pair of contours. We consider a self-similar $M$--$T$ analyzed as constant in redshift, and a constant $M$--$T$ analyzed as self-similar. In both cases the deviation from the fiducial model is very significant, with the size of constraints also affected, due to the fiducial-model assumptions having a significant impact on the number of detected clusters.  The third, and last, pair of contours show the effect of over- or underestimating the normalization mass by $40\%$ (this value agrees with what might be expected according to e.g. \citealt{Vik06}). Overestimation of the mass leads to an overestimation of $\sigma_8$, since the higher the assumed mass for a given temperature, the fewer the number of clusters at that temperature. Underestimation of the mass consequently also leads to an underestimation of $\sigma_8$.

In most cases, the constraints
are more than $3\sigma$ away from the fiducial model. Referring back to the discussion on Poisson errors in Sect.~\ref{Distriberr} and applying that to the relevant cluster distributions (see Fig.~\ref{fig:histograms}), this result is not surprising. We find that in terms of total number count Poisson error bars, the discrepancy between data and fitting assumptions are at least $\sim 6\sigma$.  These limitations will apply to any galaxy cluster survey employing cluster scaling relations to arrive at results, certainly all X-ray surveys, with the exact susceptibility to bias given by the combination of true cluster distribution and survey selection function. This stresses the importance of knowledge of the behaviour of the scaling relations in the form of self-calibration and/or separate follow-up information. For this, accurate knowledge of the selection function is necessary, so that scaling-relation scatter and evolution can be correctly distinguished \citep[as pointed out in e.g.][]{XMMLSSfirst}.

\section{Discussion and conclusions}
\label{Conclusions}

\subsection{The XCS forecast}
The XMM Cluster Survey (XCS) will cover $500$ deg$^2$ and is expected to produce one of the largest catalogues of galaxy clusters so far, with $\sim 1500$--$3300$ clusters
having $0.1\le z \le 1, 2\,{\rm keV} \le T \le 8\,{\rm keV}$. Around $20$ per cent of these will belong to the $^{500}$XCS sample that have sufficient photons ($>500$) for their X-ray temperature to be reliably estimated. In a rough approximation, we expect to find an additional $250$ or more clusters at $z>1$, of which at least $10$ should have $>500$ photons.  We have proven the potential of the XCS with the recent discovery of the most distant galaxy cluster known, XMMXCS~J2215.9-1738 at $z = 1.457$ \citep{stanford2215,hilton2215}.  Cluster redshifts are obtained from both public-domain photometry and the NOAO--XCS Survey (NXS, \citealp{nxs}). To date, more than $400$ XCS candidates have been optically confirmed. \modsec{An initial observational area of $132$ deg$^2$ (XCS DR1) contains in the range $75$--$100$ detected clusters with $T>2$~keV and $>500$ photons. This number is consistent with the theoretical expectations presented here for our fiducial models.}

We have shown the power in determining both cosmological and astrophysical parameters  expected from the {\it XMM} archive, using only self-calibration from the ($T,z$) distribution and taking detailed selection function, cluster distribution and
measurement error modeling into account in a Monte Carlo Markov Chain (MCMC) setting. Inclusion of the selection function requires the specification of the luminosity--temperature relation, and thus enables us to also self-calibrate this relation. We also introduce and motivate a new `smoothed Maximum Likelihood estimate' of the expected constraints, which can be regarded as intermediate between a Fisher matrix analysis and a full mock catalogue ensemble averaging in MCMC.

We expect the $^{500}$XCS to measure
\begin{eqnarray}
\nonumber
\begin{array}{lcrrclcrr}
  \sigma\left(\Omega_{\rm m}\right) & < & 0.03 & (10\%)\,, & & \sigma\left(\alpha\right) & < & 0.12 & (6\%)\,, \\
  \sigma\left(\sigma_8\right) & < & 0.05 & (6\%)\,, & & \sigma\left(\beta\right) & < & 0.33 & (13\%)\,,
\end{array}
\end{eqnarray}
for a flat $\Lambda$CDM Universe, the uncertainty on $\Omega_{\rm m}$ also being that
on $\Omega_\Lambda$. The cosmological constraints are similar to those already obtained using gas mass fraction measurements \citep[e.g.][]{Allen02,Allen:2007ue}.
They are better than those that can be expected from {\it XMM}--LSS \citep{RVP}, because XCS covers more area than {\it XMM}--LSS (predicted maximum area of $64$ deg$^2$, but so far results for only $5$ deg$^2$ have been published) and has a higher average exposure time. Our constraints are also somewhat competitive compared to expected constraints from e.g. the SPT, {\it Planck}, and {\it DUET} \citep{MJ04,Geisbuesch_planck}, except if self-calibration of the mass--temperature relation is also considered.  The scatter and redshift evolution of the luminosity--temperature relation cannot be jointly constrained to a significant degree by the self-calibration data alone; additional data -- archival {\it XMM} and/or follow-up -- is needed to distinguish e.g. no evolution from self-similar evolution if the scatter is left as a free parameter. Like e.g. \citet{LSW,HK,Hu}, we note that there is also potential to use this conversely, to let complementary cosmological data help constrain astrophysical parameters. We may return to this in future work.

\subsection{Measurement errors}
We include for the first time realistic temperature measurement errors, based on detailed X\textsc{spec} simulations of the {\it XMM} fields, and propagate redshift errors to the temperature determination.  The presence of realistic or worst-case measurement errors in X-ray temperature and redshift will have only a small impact on the accuracy to which cosmological parameters can be expected to be measured, of order $0.01$ in $1$D confidence limits. Furthermore, we find that imperfect knowledge of the variances of measurement errors, or the presence of catastrophic photometric redshifts, should not produce significant bias in
the cosmological constraints. We conclude that, under these assumptions, even ignoring the expected realistic measurement errors in the data analysis will provide a reasonable estimate of the true constraints.  For the case where direct $L$--$T$ data is included in the analysis, the impact of measurement errors (including susceptibility to systematics) will be larger \citep{VHS,Hu,BW}. The size of constraints forecast here provide an upper limit for that scenario.

It is already known \citep{Hut04,Hut06,LimaHuphotoz} that irreducible systematic errors in redshift estimation is a potential problem for cluster surveys, but we leave for future work the specific requirements for the XCS.

We do not yet take into account the variation of photon count with temperature/luminosity, and how that affects the size of temperature errors. Including this effect, instead of employing a lower threshold only, may well improve the size of our constraints. However the maximum improvement for self-calibration is small. For inclusion of direct $L$--$T$ data the importance will be larger.

\subsection{Cluster scaling relations}
\label{conclscalrel}
The choice of $L$--$T$ relation itself has no significant impact on the size of cosmological constraints.  In our considerations, we do not yet
take into account the separate $L$--$T$ measurement to be performed by
the XCS. In the final data analysis,
the $L$--$T$ measurements will be jointly fitted with the cluster
distribution. Hence, our expected constraints represent a worst-case
scenario of no direct data on the $L$--$T$ relation.  We plan to revisit the issue of the XCS $L$--$T$ measurement in the future. As an example, estimates for the {\it DUET} survey \citep{MJ04} show that follow-up information on the $M$--$T$ relation can improve constraints by more than a factor of three.

We quantitatively show that making incorrect assumptions about the
cluster scaling relations can typically result in at least a
$2\sigma$--$3\sigma$ bias in cosmological constraints, a result which can be considered generic for all X-ray and SZ cluster surveys, and those optical surveys relying on cluster scaling relations. Thus,
parametrizing the scaling relations appropriately and using
self-calibration and/or follow-up information is crucial to arrive at
robust results. This places high demands on precise characterization of the survey selection function to accurately distinguish scaling-relation evolution and scatter.
That is not a problem for X-ray cluster surveys (as they generally have the best-understood selection functions), and shows the importance of the XCS measurement of the $L$--$T$ relation for cosmological applications. The {\it XMM}--LSS collaboration have already pointed this out, and obtained some first results \citep{XMMLSSfirst}. A potential pitfall however is the possible redshift evolution of the $L$--$T$ scatter, as observed in the CLEF simulation \citep{Kay07}. This has not so far been considered in the literature, but is a possible source of bias that should be better understood.  The future {\it IXO} mission \citep{XEUS} will be of great importance for precision measurements of all details of the $L$--$T$ relation. The XCS will provide thousands of clusters for {\it IXO} to target.

An important source of uncertainty is the mass--temperature relation. We have shown quantitatively that, as for the luminosity--temperature relation, imperfect knowledge can easily lead to significant bias. Joint estimation of the mass--temperature relation will lead to broader constraints, and we do not expect the XCS to be able to constrain both the $L$--$T$ and $M$--$T$ relations simultaneously. Generally, it has been found that an accuracy of less than $10$ per cent in the $M$--$T$ relation will be needed, and that self-calibration (particularly if making use of the power spectrum, which XCS can not do) and/or small follow-up samples can achieve that \citep{HHM,HMH,LSW,MJ03,MJ04,Wang04,LHI,LH}.  A recent development is the claim that the X-ray luminosity is a better mass proxy than previously thought \citep{MauYX}. This remains somewhat controversial, but could be worthwhile to consider. Its potentially low scatter and the prospect of including low-temperature clusters, for which the temperature cannot be accurately measured, makes this interesting.  Likewise, employing the quantity $Y_{\rm X}$ \citep[e.g.][]{KVN} would also be an interesting option to consider. We leave the XCS-specific details for future work.

It has also been noted by, amongst others, \citet{Younger} and \citet{Ascasibar}, that the choice of parametrization for the cluster scaling relations can have a significant impact on the size of cosmological constraints, and they argue that a physically-motivated form is beneficial. As also noted by \citet{LimaHuphotoz}, efforts in correlating physical properties of clusters, such as that of \citet{Shawetal}, could therefore be of great importance for the size of cosmological constraints, not just biases or astrophysics. However, as the observed dependence on parametrization appears to largely come from an $\Omega_{\rm m}$--$\Omega_{\Lambda}$ degeneracy, and in this work we assume that $\Omega_{\Lambda}=1-\Omega_{\rm m}$, we do not expect this to be of importance for our results here.

\modsec{\subsection{Selection function}
The cluster model assumptions made in our calculation of the selection function could have an impact on cluster detectability and thus cosmological constraints. For this reason, we have studied the selection function dependence on cluster structure parameters (for the beta model assumed). For core radii within reasonable bounds, the relative difference compared to our standard beta model is of the order $10$ per cent.
This number is however an overestimation to the resulting overall uncertainty in cluster number predictions, as it does not include any detailed knowledge of the cluster population, in particular the detailed distribution of cluster-model parameters. Therefore, the significance is limited. We expect the real differences to be smaller. These results agree with those of \citet{Burenin07} for the 400d survey, which show that reasonable variations in cluster size, morphology and scaling relations induce an uncertainty in the detectability for a given flux of typically less than $5$ per cent.  In a future, real analysis, the cluster-model parameter distribution would be included along with a selection function dependence on these parameters, thereby also significantly reducing any uncertainty.  We are currently studying the selection function further in this respect. Among other things we are applying the detection pipeline to the simulated clusters from CLEF \citep{Kay07}, to understand the impact of various cluster properties (Hosmer et al., in preparation).}

\subsection{Sensitivity to fiducial model and priors}
\modsec{
Dropping our assumption of spatial flatness, we do not expect particularly strong constraining power on $\Omega_{\Lambda}$, based on results such as \cite{Allen02,MJ04}. One should note though that the utilization of X-ray galaxy clusters over a large range of redshifts has additional constraining power compared to \citet{Allen02}, as those results are based only on the gas mass fraction in nearby clusters; most of the constraining power on $\Omega_{\Lambda}$ comes from $z>0.5$ \citep{HHM,LSW}. The constraints on $\Omega_{\rm m}$ and $\sigma_8$ should broaden as a consequence of dropping the flatness assumption; however \citet{MJ04} find that such increase tends to be fairly marginal.  This increase could arguably also be alleviated by employing appropriate parametrizations for the cluster scaling relations, as discussed above in Sect.~\ref{conclscalrel}. This, as well as the constraining power on modified-gravity models, is a topic for further investigation.}

The assumption of fixed values for the priors of some cosmological parameters,
e.g. the scalar spectral index and the Hubble constant, is not realistic given the uncertainty that still exists regarding their true values. Relaxing those priors would increase the size of constraints, though probably not in a significant manner.

\modsec{\subsection{Requirements for a real analysis}
Our methodology contains a number of simplifications that are sufficient for a forecast, but for a future real analysis will need to be modified.
In particular, this applies to the calculation of the matter-field dispersion $\sigma$, which ought to be calculated employing e.g.~\citet{EH98} or \textsc{CMBFAST} \citep{CAMB} for the transfer function.  The uncertainty in the selection function will have to be better understood and included explicitly. In this context, understanding the distribution and impact of cluster structure parameters and their correlation is important.  New mass function fits which are more accurate than \citet{jetal} should also be used, and the impact of uncertainty included.  Uncertainty in the $M$--$T$ normalization chosen will have to be more carefully considered, as will possible systematic uncertainties in temperature definitions and in photometric redshifts.  We will in the final analysis fit the $L$--$T$ relation jointly with cosmology (i.e. use flux data as well), thereby achieving some improvement on the constraints presented here.}

\subsection{The role of XCS, future surveys and outlook}
\modsec{The XCS is a complete investigation into the galaxy cluster content of the {\it XMM} archive. It will be a pathfinder survey for the many ongoing and planned galaxy cluster surveys, and in particular help guide design of new {\it XMM} observations and the upcoming X-ray missions {\it eROSITA} and {\it IXO}. The $L$--$T$ measurement from the XCS will be the most accurate so far, and will provide an important calibration for potentially all cluster surveys. Although the expected cosmological constraints may not be particularly tight compared to those from the combination of other cosmological data, they will provide important complementary, independent tests of the standard cosmological model from the largest X-ray cluster sample. Particularly, we expect to provide good independent constraints on $\sigma_8$.  }

For example, re-imaging of the XCS sample with {\it XMM} and/or {\it IXO} could improve temperature errors sufficiently, so that at least some of the remaining $\sim 80$ per cent of the XCS clusters not in the $^{500}$XCS could be used for constraints (this corresponds to no photon-count cut-off in our calculations -- effectively a $\sim 50$-photon cut-off). We find that an upper limit on the improvement in constraints is by $1\sigma$ ($2$D) or $\sim 40$ per cent ($1$D), which thus also represents the best one could possibly do with the current {\it XMM} archive \emph{using $(T, z)$ self-calibration only}.  Once we add direct $L$--$T$ data to the procedure, the lever arm from re-imaging will be more significant. The {\it XMM}--LSS collaboration argue \citep{XMMLSSfirst} that the most efficient way to constrain the $L$--$T$ evolution is to increase the sample size, rather than improve temperature errors, and propose a future $200$ deg$^2$ {\it XMM} survey with this rationale \citep{XMMLSSfuture}.  A complementary approach to additional observations would be to also use the luminosity--mass relation as mass proxy for those clusters for which the temperature determination is difficult \citep{MauYX}, or use the relatively new quantity $Y_{\rm X}$ advocated by e.g. \citet{KVN}.

The future for galaxy clusters as a precision and complementary cosmological tool looks increasingly promising, with a range of surveys planned or underway, and numerous simulations undertaken to understand the mass function and cluster physics.  The XCS will produce one of the largest ever catalogues of galaxy clusters, providing valuable information on cosmology and cluster physics through the luminosity--temperature relation, beating a path for the many planned surveys.  The interface between well-understood cluster physics and cosmology, cross-calibration, and complementary cosmological data will surely be important for constraining dark energy, the primordial power spectrum, and cluster physics over the next decade and beyond.

\section*{ACKNOWLEDGMENTS}

M.S.\ was partially supported by the Swedish G{\aa}l\"{o}, Gunvor~\& Josef An\'er and
C.~E.~Levin foundations and the Sir Richard Stapley Educational Trust.  P.T.P.V.\ acknowledges the support of POCI2010 through the project POCTI/CTE-AST/58888/2004.
A.K.R., E.L.-D., M.H.~\& N.M.\ were supported by PPARC/STFC, and M.D. in part by the RAS Hosie Bequest.  M.S., M.D., K.S.~\& M.H.\ acknowledge additional financial support from their respective universities.  A.K.R~\& K.S.\ acknowledge financial support from the {\it XMM} and {\it Chandra} guest observer programmes and from the NASA--LTSA award NAG-11634.  The work by S.A.S.\ was performed under the auspices of the
U.S. Department of Energy, National Nuclear Security Administration by
the University of California, Lawrence Livermore National Laboratory
under contract No.~W-7405-Eng-48.  This work is based on data obtained by {\it XMM--Newton}, an ESA science mission funded by contributions from ESA member states and from NASA.  The research made use of the NASA/GSFC-supported \textsc{Xspec} software, and was conducted in co-operation with SGI/Intel utilizing the Altix 4700 supercomputer (COSMOS), also funded by HEFCE and STFC.  We are grateful to Ronald Cools and the NINES
group at KU~Leuven for making the numerical integration package
\textsc{CUBPACK} available to us.  We thank Ben Maughan and Jochen Weller for useful
discussions.



\appendix
\onecolumn

\section{Equations}
\label{Appendix}

\subsection{Cluster counts}
\label{clustcnt}
\subsubsection{Ideal measurements}
The expected number of clusters with temperatures between $T_1$ and
$T_2$ at redshifts between $z_1$ and $z_2$ when measurements are
assumed to be exact is given by
\begin{eqnarray}
  N_{\rm ideal}(T_1, T_2, z_1, z_2) & = & \int_{z_1}^{z_2}
  \int_{T_1}^{T_2}
  n_{\rm ideal}(T, z) {\rm d}T {\rm d}z
\end{eqnarray}
where $n_{\rm ideal}$ is the actual number density of clusters in
temperature and redshift, given by the convolution of the mass
function $n\left(M_{\rm t}, z\right)$ with cluster scaling relations, their
scatter (through $p\left(L_{\rm t}, M_{\rm t}\right)$), cosmic volume ${\rm d}V/{\rm d}z$ and
the survey selection function $f_{\rm sky}$ (including sky coverage):
\begin{eqnarray}
  n_{\rm ideal}(T, z) & = &
  \int_{M_{\rm t}} \int_{L_{\rm t}}
   n\left(M_{\rm t}, z\right) f_{\rm sky}(L_{\rm t}, T, z) p\left(L_{\rm t}, M_{\rm t} | L(T,
  z), M(T, z)\right)
   \frac{{\rm d}V}{{\rm d}z} {\rm d}L_{\rm t} {\rm d}M_{\rm t} \,.
\end{eqnarray}
The scaling-relation scatter probability distributions are assumed to
be statistically independent,
\begin{eqnarray}
    p\left(L_{\rm t}, M_{\rm t} | L(T, z), M(T, z)\right) & = & p\left(L_{\rm t} | L(T,
    z)\right)
    \times p\left(M_{\rm t} | M(T, z)\right) \,,
\end{eqnarray}
each having a log-normal form:
\begin{eqnarray}
  \nonumber
  p\left(M_{\rm t} | M(T,z), T, z \right) {\rm d}M_{\rm t} & = & p\left(T^M_{\rm t}(M_{\rm t}) | T,
  z\right) \frac{{\rm d}M_{\rm t}}{{\rm d}T^M_{\rm t}} {\rm d}T^M_{\rm t} =
   \frac{1}{{\rm erf}(m_T/\sqrt{2})\sqrt{2\pi}\sigma_{\log T}}
   \exp\left[-\frac{1}{2}\frac{\left(\log_{10}T -
  \log_{10}T^M_{\rm t}\right)^2}{\sigma^2_{\log T}}\right] \\
   & & \times \Theta\left(m_T\sigma_{\log T} - |\log_{10}T -
  \log_{10}T^M_{\rm t}|\right) \frac{{\rm d}M_{\rm t}}{{\rm d}T^M_{\rm t}} {\rm d}\log_{10}T^M_{\rm t} \,,
 \\
 \nonumber
 p\left(L_{\rm t} | L(T,z)\right) {\rm d}L_{\rm t} & = &
    \frac{1}{{\rm erf}(m_L/\sqrt{2})\sqrt{2\pi}\sigma_{\log L_{\rm X}}}
   \exp\left[-\frac{1}{2}\frac{\left(\log_{10}L(T,z)
    - \log_{10}L_{\rm t}\right)^2}{\sigma^2_{\log L_{\rm X}}}\right] \\
    & & \times \Theta\left(m_L\sigma_{\log L_{\rm X}} -
  |\log_{10}L(T,z) - \log_{10}L_{\rm t}|\right)
    {\rm d}\log_{10}L_{\rm t} \,.
\end{eqnarray}
The parameters $m_T$, $m_L$, $\sigma_{\log T}$ and $\sigma_{\log L}$
are described further in Sections \ref{tempmass} and \ref{lumtemp} as
well as Table \ref{tab:scalrel}.

\subsubsection{Measurement errors}
\label{clustcnterr}
When treating the case of measurement errors in $T$ and $z$, we must
distinguish observed and true temperature. The expected number of
clusters between \emph{observed} temperatures $T_1$ and $T_2$, and
redshifts $z_1$ and $z_2$, is given by
\begin{eqnarray}
  N_{\rm obs}(T_1, T_2, z_1, z_2) & = & \int_{z_1}^{z_2} \int_{T_1}^{T_2}
  \overline{n}(T, z) {\rm d}T {\rm d}z
\end{eqnarray}
where $\overline{n}$ represents the cluster distribution marginalized
over the probability distribution for measurements, i.e.
\begin{eqnarray}
  \nonumber
  \overline{n}(T, z) & = &
  \int_{z_{\rm t}} \int_{T_{\rm t}}
  n_{\rm ideal}(T_{\rm t}, z_{\rm t}) p\left(T, z | T_{\rm t}, z_{\rm t}\right) {\rm d}T_{\rm t} {\rm d}z_{\rm t}\\
 & = & \int_{z_{\rm t}} \int_{T_{\rm t}}
  n_{\rm ideal}(T_{\rm t}, z_{\rm t}) p\left(\left. T\left[ \frac{1+z_{\rm t}}{1+z}
  \right]  \right| T_{\rm t}, z_{\rm t}\right)p(z |
  z_{\rm t})\left(\frac{1+z_{\rm t}}{1+z}\right) {\rm d}T_{\rm t} {\rm d}z_{\rm t} \,,
\end{eqnarray}
where $z_{\rm t}$ and $T_{\rm t}$ are true redshift and temperature, and in the
last step the relation
$T_{\rm obs} =
(1+z_{\rm obs})T_{\rm t}/(1+z_{\rm t})$
was used to go from observed to true
temperature.  The temperature and redshift measurement probability
distributions are modelled by
\begin{eqnarray}
    \label{eq:tdist}
     p\left(T | T_{\rm t}, z_{\rm t}\right) {\rm d}T
    & = &
    \frac{1}{\sqrt{\pi/2}\left(\sigma_T^-+\sigma_T^+\right)}
    \exp\left[-\frac{1}{2}\frac{\left(T-T_{\rm med}(T_{\rm t},
     z_{\rm t})\right)^2}{\sigma_{T}(T_{\rm t}, z_{\rm t})^2}\right] {\rm d}T
   \\
    \nonumber
   & & T_{\rm med}(T_{\rm t}, z_{\rm t})/T_{\rm t} = \alpha_c + \alpha_T T_{\rm t} + \alpha_z
     z_{\rm t} + \alpha_{zz}z_{\rm t}^2 + \alpha_{TT} T_{\rm t}^2 +
   \alpha_{zT} z_{\rm t} T_{\rm t} \\
    \nonumber
   & & \sigma_{T}(T_{\rm t}, z_{\rm t}) =
\left\{
  \begin{array}{ll}
    \sigma_T^+ = T_{\rm t}\left(\beta^{+}_c + \beta^{+}_T T_{\rm t} + \beta^{+}_z
     z_{\rm t} + \beta^{+}_{zz}z_{\rm t}^2 + \beta^{+}_{TT} T_{\rm t}^2 + \beta^{+}_{zT}
     z_{\rm t} T_{\rm t}\right), & \hbox{$T_{\rm t} \geq T_{\rm med}(T_{\rm t}, z_{\rm t})$;} \\
    \sigma_T^- = T_{\rm t}\left(\beta^{-}_c + \beta^{-}_T T_{\rm t} + \beta^{-}_z
     z_{\rm t} + \beta^{-}_{zz}z_{\rm t}^2 + \beta^{-}_{TT} T_{\rm t}^2 + \beta^{-}_{zT}
     z_{\rm t} T_{\rm t}\right), & \hbox{otherwise.}
  \end{array}
\right.
   \\
    \nonumber
    & & \mathrm{where\,the\,\alpha \,and\,\beta
     \,are\,determined\,from\,simulations\,(see\,Sect.\,\ref{Terror})}
    \\
    p\left(z | z_{\rm t}\right) {\rm d}z & = &
   \frac{1}{N^{\rm rand}_z(z_{\rm t})(1-f_{\rm cat}) + N_z^{\rm
     cata}(z_{\rm t})f_{\rm cat}}\left\{(1-f_{\rm cat})
     \exp\left[-\frac{1}{2}\frac{\left(z-z_{\rm t}\right)^2}{\sigma_0^2(1+z_{\rm t})^2}\right]\right.
     + \\
   \nonumber
   & & \left. f_{\rm cat}
     \exp\left[-\frac{1}{2}\frac{\left(z-z_{\rm t}\right)^2}{c^2\sigma_0^2(1+z_{\rm t})^2}\right]
   \Theta\left(|z-z_{\rm t}| - nc\sigma_0(1+z_{\rm t})\right) \Theta(z)\right\}
    {\rm d}z
    \\
   \nonumber
    & & N^{\rm rand}_z(z_{\rm t}) = \sqrt{\frac{\pi}{2}}\sigma_0(1+z_{\rm t})\left[
    1 +
    {\rm erf}\left(\frac{z_{\rm t}}{\sqrt{2}\sigma_0(1+z_{\rm t})}\right)\right] \\
   \nonumber
    & & N_z^{\rm cat}(z_{\rm t}) =
    \sqrt{2\pi}c\sigma_0(1+z_{\rm t})\left(
    {\rm erfc}\left(\frac{n}{\sqrt{2}}\right) -
         \frac{1}{2} \min\left[{\rm
     erfc}\left(\frac{z_{\rm t}}{\sqrt{2}c\sigma_0(1+z_{\rm t})}\right),
        {\rm erfc}\left(\frac{n}{\sqrt{2}}\right)\right]\right)
\end{eqnarray}
The temperature error assumptions made in this work are described further in Sect.~\ref{Terror} and summarized in Table~\ref{tab:terrs}.
The redshift error assumptions for parameters $f_{\rm cat}$, $n$, $c$ and $\sigma_0$ are described
further in Section \ref{Photozs} and Table \ref{tab:zerror}.
Note
that the probability distributions of \emph{true} temperatures and
redshifts, the Bayesian "inverses" of the above, are weighted by the
cluster distribution and given by

\begin{eqnarray}
    p\left(T_{\rm t} | T, z_{\rm t}\right) {\rm d}T_{\rm t} & = &
    \frac{p\left(T | T_{\rm t}, z_{\rm t}\right)n_{\rm ideal}(T_{\rm t}, z_{\rm t}) {\rm d}T_{\rm t}}
    {\int p\left(T | T', z_{\rm t}\right)n_{\rm ideal}(T', z_{\rm t}) {\rm d}T'} \,,
   \\
    p\left(z_{\rm t} | T_{\rm t}, z\right) {\rm d}z_{\rm t} & = &
    \frac{p\left(z | z_{\rm t}\right)n_{\rm ideal}(T_{\rm t}, z_{\rm t}) {\rm d}z_{\rm t}}
    {\int p\left(z | z'\right)n_{\rm ideal}(T_{\rm t}, z') {\rm d}z'} \,.
\end{eqnarray}

\begin{table}
\centering
\begin{tabular}{|c|c|}
   \hline
  Realistic $T$ errors & Worst-case $T$ errors \\
  \hline
  Eq.~(\ref{eq:tdist}) & Eq.~(\ref{eq:tdist}) \\
  & with std. dev. $3\times\sigma_T$ \\
  \hline
\end{tabular}
\caption{Temperature error specifications.}
\label{tab:terrs}
\end{table}

\begin{table}
\centering
\begin{tabular}{|l|c|c|c|c|c|}
   \hline
  Parameter & Description & Realistic $z$ errors & Worst-case $z$ errors \\
  \hline
  $\sigma_0$ & Standard deviation at $z=0$ & 0.05 & 0.10 \\
  $c$ & Catastrophic standard deviation in units of $\sigma_0$ & 4 & 4\\
  $n$ & Min. deviation from mean in units of $c\sigma_0$ for
  catastrophic redshifts & 1 & 1\\
  $f_{\rm cat}$ & Fraction of catastrophic redshifts & 0.05 & 0.10 \\
  \hline
\end{tabular}
\caption{Redshift error specifications.}
\label{tab:zerror}
\end{table}

\subsection{Expected likelihood}
\label{explik}

In order to evaluate the expected constraints from a survey, one needs to consider some ensemble of possible outcomes and from that calculate, by ensemble averaging or otherwise (given a specification of `expected'), the expected constraints. We have chosen a type of smoothed Maximum Likelihood (ML) estimate, that captures the most likely shape and size of constraint contours but removes the offset associated with a traditional ML point estimate. In the following we show in detail that our expected constraints can be obtained accurately without averaging over many data realizations, but rather by using only an `average catalogue'.

Having an expression for the single-catalogue likelihood, we seek to estimate the
expected constraints for the survey. We define this as the expected
constraints for a set consisting of a certain fraction $\varepsilon$
most likely catalogues. We start by setting up some formalism and prove our central theorem, and then go on to use this for our application.
\begin{defn}
Let $\{C_j\}$ denote a set of catalogues indexed by $j$. Let $N$ be the number of bins of a catalogue. Let $N_i$ or $N_i^j$ be the observed number count for bin $i$, in catalogue $j$ where superscript present. Let $\lambda_i$ be the Poisson mean for bin $i$ at which the likelihood is evaluated, and $\lambda_i^*$ the same for the fiducial model used to generate the catalogues. Let $\delta_i^j \equiv N_i^j - \lambda_i^*$
measure the deviation of the observed number count from the
fiducial-model mean.
\end{defn}
\begin{defn}
Let the expected likelihood for the fraction $\varepsilon$
most likely catalogues in a Poisson ensemble be given by
\begin{equation}
\left\langle{\cal L}\right\rangle_{\varepsilon}\, \equiv \,
\prod_{i}e^{-\lambda_i}\left\langle\prod_{i}
\left[\frac{\lambda_i^{N_i}}{N_i!}
\right]\right\rangle_{\varepsilon}\,,
\end{equation}
where the product runs over the $N$ bins in a catalogue, and $\langle\cdot\rangle_{\varepsilon}$
denotes a
Poisson ensemble average restricted to catalogues $C_j$ such that
$\sum_j P(C_j)\Theta(P(C_j)-P_{\varepsilon}) = \varepsilon$ (with
$\Theta$ the Heaviside step function). This expression also defines
the probability threshold $P_{\varepsilon}$.
\end{defn}
\begin{corrol}
It follows from the above definition and the Poisson distribution that
\begin{eqnarray}
\left\langle{\cal L}\right\rangle_{\varepsilon}\, = \,
\prod_{i}e^{-\lambda_i}
\sum_j \frac{P(C_j)}{\varepsilon}\prod_{i}
\left[\frac{\lambda_i^{N_i^j}}{N_i^j!}
\right] \Theta(P(C_j)-P_{\varepsilon})\, = \,
e^{-\sum_i(\lambda_i+\lambda_i^*)}
\frac{1}{\varepsilon}\sum_j
\prod_{i}\frac{(\lambda_i\lambda_i^*)^{\lambda_i^*+\delta_i^j}}
{\left[(\lambda_i^*+\delta_i^j)!\right]^2}
\Theta(P(C_j)-P_{\varepsilon})\,.
\end{eqnarray}
\end{corrol}

\begin{defn}
Let
\begin{equation}
\mathcal{C}_{\pm} \equiv \left\{ \left\{\delta_i\right\}_{i=1}^N \left|
 \delta_i \in \{ \lceil\lambda_i^*\rceil-\lambda_i^*,
 \lfloor\lambda_i^*\rfloor-\lambda_i^* \}\right.
 \forall i \right\}\,,
\end{equation}
the set of catalogues consisting of the $2^{N}$ catalogues between the most likely
catalogue (for which $\delta_i = \lfloor\lambda_i^*\rfloor-\lambda_i^*
\,\forall i$) to the catalogue with probability $P_{\varepsilon}$ (for
which $\delta_i = \lceil\lambda_i^*\rceil-\lambda_i^*\,\forall
i$). Here, $\lceil\cdot\rceil$ and $\lfloor\cdot\rfloor$ are the
ceiling and floor operators respectively.
\end{defn}
\begin{rem}
The choice of this set of catalogues will be convenient and is suitable to define a smoothed ML estimate.
\end{rem}

We now come to the central theorem:
\begin{thm}
For the catalogue set $\mathcal{C}_{\pm}$,
\begin{equation}
\left\langle{\cal L}\right\rangle_{\varepsilon} = \sum_i\left(\lceil\lambda_i^*\rceil - \frac{1}{2} \right)\ln
\lambda_i -\lambda_i + \mathcal{O}\left(\delta^3\right) + {\rm const.}
\end{equation}
\end{thm}

\begin{proof}
The probability level
$\varepsilon$ for the catalogue set $\mathcal{C}_{\pm}$ can be estimated through
\begin{equation}
\prod_i \frac{\left(\lambda_i^*\right)^{\lceil\lambda_i^*\rceil}}
{\lceil\lambda_i^*\rceil!}
  \leq \frac{\varepsilon}{2^Ne^{-\sum_i \lambda_i^*}} \leq
\prod_{i}\frac{\left(\lambda_i^*\right)^{\lfloor\lambda_i^*\rfloor}}
{\lfloor\lambda_i^*\rfloor!}\,.
\end{equation}
We approximate
\begin{equation}
\varepsilon \approx 2^Ne^{-\sum_i
\lambda_i^*}\prod_{i}\frac{\left(\lambda_i^*\right)^{\lambda_i^*}}
{\Gamma\left(1+\lambda_i^*\right)}\,,
\end{equation}
where we have used the gamma function as a continuation of the
factorial, effectively extending the Poisson distribution to the gamma
distribution for non-integer values of $N_i$, something we will use
throughout.  We can now write
\begin{eqnarray}
\left\langle{\cal L}\right\rangle_{\varepsilon}\, & = &
2^{-N}e^{-\sum_i \lambda_i}
\sum_j \prod_{i}\frac{(\lambda_i)^{\lambda_i^*+\delta_i^j}
(\lambda_i^*)^{\delta_i^j}}{\left[(\lambda_i^*+\delta_i^j)!\right]^2}
\Gamma\left(1+\lambda_i^*\right)\,,
\end{eqnarray}
where the catalogues (indexed by $j$) are now restricted to those in
$\mathcal{C}_{\pm}$.  To proceed, we first take the logarithm of the
likelihood to separate out the catalogue-set-dependent normalization,
which is of no consequence for our discussion.  We can thus write
\begin{eqnarray}
\ln \left\langle{\cal L}\right\rangle_{\varepsilon}\, & = &
-N\ln 2
+\sum_i\left[-\lambda_i + \lambda_i^*\ln \lambda_i +
\ln\Gamma\left(1+\lambda_i^*\right) \right] +\ln \widehat{\Sigma} \,,
\end{eqnarray}
where we have defined
\begin{equation}
\widehat{\Sigma} \equiv \sum_j
\prod_{i}\frac{\left(\lambda_i\lambda_i^*\right)^{\delta_i^j}}
{\left[\left(\lambda_i^*+\delta_i^j\right)!\right]^2}\,.
\end{equation}
Taylor expanding in $\delta_i^j$ (since $|\delta_i^j|<1$ for our
catalogues) we find
\begin{eqnarray}
\nonumber
\ln \left\langle{\cal L}\right\rangle_{\varepsilon}\, & = &
-N\ln 2
+\sum_i\left[\lambda_i^*\ln \lambda_i -\lambda_i +
\ln\Gamma\left(1+\lambda_i^*\right)\right]
+\ln \left.\widehat{\Sigma}\right|_{\delta=0}
+\sum_{i,j}
\left.\left(\frac{1}{\widehat{\Sigma}}\frac{d\widehat{\Sigma}}{d\delta_i^j}
\right)\right|_{\delta=0} \delta_i^j + \\
& &  \frac{1}{2}\sum_{i,j,k,l} \left.\left[\frac{1}{\widehat{\Sigma}}
\left(\frac{d^2\widehat{\Sigma}}{d\delta_i^j d\delta_k^l } -
\frac{1}{\widehat{\Sigma}}
\frac{d\widehat{\Sigma}}{d\delta_i^j}\frac{d\widehat{\Sigma}}{d\delta_k^l}
\right)
 \right]\right|_{\delta=0} \delta_i^j\delta_k^l +
\mathcal{O}\left(\delta^3\right)\,,
\end{eqnarray}
where ``$\delta=0$'' denotes $\delta_i^j=0\,\forall i,j$. Inserting
$\widehat{\Sigma}$ and the derivatives
\begin{eqnarray}
\frac{d\widehat{\Sigma}}{d\delta_i^j} & = & \ell(\lambda_i,
\lambda_i^*)
\prod_{k}\frac{(\lambda_k\lambda_k^*)^{\delta_k^j}}
{\left[(\lambda_k^*+\delta_k^j)!\right]^2}
\,, \\
\frac{d^2\widehat{\Sigma}}{d\delta_i^j d\delta_k^l } & = &
\ell(\lambda_i, \lambda_i^*) \ell(\lambda_k, \lambda_k^*)
\prod_{m}\frac{(\lambda_m\lambda_m^*)^{\delta_m^j}}{\left[
(\lambda_m^*+\delta_m^j)!\right]^2}
 \tilde{\delta}_{jl}\,,
\end{eqnarray}
where $\ell\left(\lambda,\lambda^*\right) \equiv \ln \left(
\lambda\lambda^* \right) - 2\Psi \left( 1+\lambda^* \right)$ (the
digamma function $\Psi$ coming from the factorial as gamma function),
we obtain
\begin{eqnarray}
\nonumber
\ln \left\langle{\cal L}\right\rangle_{\varepsilon}\, & = &
\sum_i\left(\lambda_i^*\ln \lambda_i -\lambda_i\right)
+\sum_{i} \ell(\lambda_i, \lambda_i^*) 2^{-N}\sum_j \delta_i^j +
\\ & & \frac{1}{2}\sum_{i,j,k,l} 2^{-N}
\left[\ell(\lambda_i, \lambda_i^*)\ell(\lambda_k, \lambda_k^*)
\tilde{\delta}_{jl}
          - 2^{-N}\ell(\lambda_i, \lambda_i^*)\ell(\lambda_k, \lambda_k^*)
\right]
\delta_i^j\delta_k^l + \mathcal{O}\left(\delta^3\right)\\
\nonumber
& = &
\sum_i\left(\lambda_i^*\ln \lambda_i -\lambda_i\right)
+\sum_{i} \ell(\lambda_i, \lambda_i^*) 2^{-N}\sum_j \delta_i^j +
\\
\label{lnavl}
& &
\frac{1}{2}\sum_{i,k} \ell(\lambda_i, \lambda_i^*)\ell(\lambda_k,
\lambda_k^*) \left[ 2^{-N} \sum_{j}\delta_i^j\delta_k^j -
 2^{-2N} \sum_{j,l} \delta_i^j\delta_k^l \right] +
\mathcal{O}\left(\delta^3\right)\,,
\end{eqnarray}
where $\tilde{\delta}_{ij}$ is the Kronecker delta.  We can evaluate
the $\delta$-sums using our knowledge of the set of catalogues
$\mathcal{C}_{\pm}$:
\begin{eqnarray}
\label{deltasumsing}
\sum_j {\delta_i^j} & = & 2^{N-1}\left(\lceil\lambda_i^*\rceil +
\lfloor\lambda_i^*\rfloor-2\lambda_i^*\right) =
2^N\left(\Delta_i^* - \frac{1}{2} \right)\,, \\
\label{deltasumdoub}
\sum_{j,l} {\delta_i^j\delta_k^l} & = & 2^{2N}\left(\Delta_i^* -
\frac{1}{2} \right)
\left(\Delta_k^* - \frac{1}{2} \right) = 2^{2N}\left[
\Delta_i^*\Delta_k^* - \frac{1}{2}\left(\Delta_i^*+\Delta_k^*\right) +
\frac{1}{4}
\right] \,,
\\
\label{deltasumdoubhalf}
\sum_{j} {\delta_i^j\delta_k^j} & = &
\frac{2^N}{4}\left[
\Delta_i^*\Delta_k^* + \Delta_i^*\left(\Delta_k^*-1\right)
+\left(\Delta_i^*-1\right)\Delta_k^*+\left(\Delta_i^*-1\right)
\left(\Delta_k^*-1\right)
\right] = 2^N\left[\Delta_i^*\Delta_k^* -\frac{1}{2}
\left(\Delta_i^*+\Delta_k^*\right)
+\frac{1}{4}\right]\,,
\end{eqnarray}
where we have defined $\Delta_{i}^* \equiv \lceil\lambda_i^*\rceil -
\lambda_i^*$ and excluded the possibility that
$\lceil\lambda_k^*\rceil = \lfloor\lambda_k^*\rfloor =
\lambda_k^*$. Inserting (\ref{deltasumdoub}) and
(\ref{deltasumdoubhalf}) in (\ref{lnavl}) we find that the
second-order term is zero due to cancellation between its two
constituent terms. Hence, also inserting (\ref{deltasumsing}), we
finally arrive at
\begin{eqnarray}
\ln \left\langle{\cal L}\right\rangle_{\varepsilon}\, & = &
\sum_i\left[\lambda_i^*\ln \lambda_i -\lambda_i+
\ell(\lambda_i, \lambda_i^*)\left(\Delta_i^* - \frac{1}{2}
\right)\right]  + \mathcal{O}\left(\delta^3\right)\\
\label{lnleps}
 & = &
\sum_i\left[\left(\lceil\lambda_i^*\rceil - \frac{1}{2} \right)\ln
\lambda_i -\lambda_i+
\left(\Delta_i^* - \frac{1}{2} \right)\left(\ln \lambda_i^* -
2\Psi(1+\lambda_i^*)\right)\right]  +
\mathcal{O}\left(\delta^3\right)\,.
\end{eqnarray}
\end{proof}

The theorem states that a good approximation to $\left\langle{\cal
L}\right\rangle_{\varepsilon}$ is given by using $N_i =
\lceil\lambda_i^*\rceil - 1/2$ in a single-catalogue likelihood $\cal
L$.  This expression, however, does give rise to an offset in the
best-fitting values away from the true means, associated with shot noise. As we
are using the catalogue construction as a way of defining a meaningful
expected likelihood which is not just an arbitrary point estimate, we
are not really interested in this offset (and would like to separate it from sources of bias); rather the variance is what
concerns us. Therefore, we propose using the very similar expression
\begin{eqnarray}
\label{meanlnl}
\left\langle \ln{\cal L}\right\rangle\, & = &
\sum_i\left(\lambda_i^*\ln \lambda_i -\lambda_i\right) + {\rm const.}
\end{eqnarray}
The best-fitting values for $\lambda_i$ of this expression are equal to
the true means $\lambda_i^*$.  However, how do the standard deviations
compare?  The standard deviations are given by
\begin{equation}
\sigma_{\varepsilon,i} = \sqrt{\lceil\lambda_i^*\rceil - 1/2}\,, \qquad
\sigma_{{\rm mean},i} = \sqrt{\lambda_i^*}\,,
\end{equation}
where $\sigma_{\varepsilon,i}$ is the standard deviation of
equation~(\ref{lnleps}) and $\sigma_{{\rm mean},i}$ the standard deviation
of equation~(\ref{meanlnl}). Upper and lower limits for their ratio can
then be given as
\begin{eqnarray}
\frac{1}{\sqrt{1+1/2\lambda_i^*}} < \frac{\sigma_{{\rm
mean},i}}{\sigma_{\varepsilon,i}} <
\frac{1}{\sqrt{1-1/2\lambda_i^*}}\,.
\end{eqnarray}
It is clear that for $\lambda_i^*<1$ the relative error will become
large as $\lambda_i^*$ decreases. Again, this is due to shot
noise. One could always make bins large enough that at least a few
elements fall in each bin, ensuring only moderate relative errors in
the standard deviations. Such a binning might however not be optimal
or even close to, and thus reflect the underlying distribution poorly.
It appears that no general conclusion can be drawn here. However, if
we specify a dependence $\lambda_i=\lambda_i^*(\theta/\theta^*)^{a_i}$
for the $\lambda_i$'s on some parameter $\theta$, as is typically the
case and certainly here, we can write the following:
\begin{eqnarray}
 \left\langle\ln{\cal L}\right\rangle\, =
\sum_i\left(\lambda_i^*\ln \lambda_i -\lambda_i\right) + {\rm const.}
 =  \ln \theta \sum_i a_i \lambda_i^*  - \sum_i
\lambda_i^*\left(\frac{\theta}{\theta^*}\right)^{a_i} + {\rm const.}
\end{eqnarray}
\begin{eqnarray}
\ln \left\langle{\cal L}\right\rangle_{\varepsilon}\,  =
\sum_i\left[\left(\lceil\lambda_i^*\rceil-\frac{1}{2}\right)\ln
\lambda_i -\lambda_i\right] + {\rm const.}
 =  \ln \theta\sum_i
a_i\left(\lceil\lambda_i^*\rceil-\frac{1}{2}\right) - \sum_i
\lambda_i^*\left(\frac{\theta}{\theta^*}\right)^{a_i} + {\rm const.}
\end{eqnarray}
Clearly, the only difference between $\ln \left\langle{\cal
L}\right\rangle_{\varepsilon}$ and $\left\langle\ln{\cal
L}\right\rangle$ comes from the difference in the first sum. Naively,
we would not expect this to differ much between the two cases,
particularly for a binning that represents the distribution well. What
would be the expected value? Consider the following quantity:
\begin{eqnarray}
s_{\rm rel} \equiv \frac{\sum_i a_i\lambda_i^*}{\sum_i
a_i\left(\lceil\lambda_i^*\rceil-\frac{1}{2}\right)} =
\frac{\sum_i a_i \lambda_i^*}{\sum_i a_i \left(\lambda_i^* +
(\lceil\lambda_i^*\rceil-\lambda_i^*)-\frac{1}{2}\right)} \,.
\end{eqnarray}
One would generally expect that $(\lceil\lambda^*\rceil-\lambda^*) \in
U(0,1)$ or at least a similarly symmetric distribution across the
bins, so that $\left\langle
\lceil\lambda^*\rceil-\lambda^*\right\rangle = 1/2$. We thus expect
\begin{eqnarray}
\left\langle s_{\rm rel} \right\rangle =
\frac{\sum_i a_i \lambda_i^*}{\left\langle\sum_i
a_i\left(\lceil\lambda_i^*\rceil-\frac{1}{2}\right)\right\rangle} =
\frac{\sum_i a_i \lambda_i^*}{\sum_i a_i \lambda_i^*} = 1 \,.
\end{eqnarray}
For typical XCS catalogues, even if we assign uncorrelated random
exponents $a_i$, the probability distribution for $s_{\rm rel}$ is
quite generally very sharply peaked at or close to $s_{\rm rel}=1$. An
example is shown in Fig.~\ref{fig:srel}, for which $a_i \in U(-5,5)$.
\begin{figure}
\begin{center}
\includegraphics[width=0.5\linewidth]{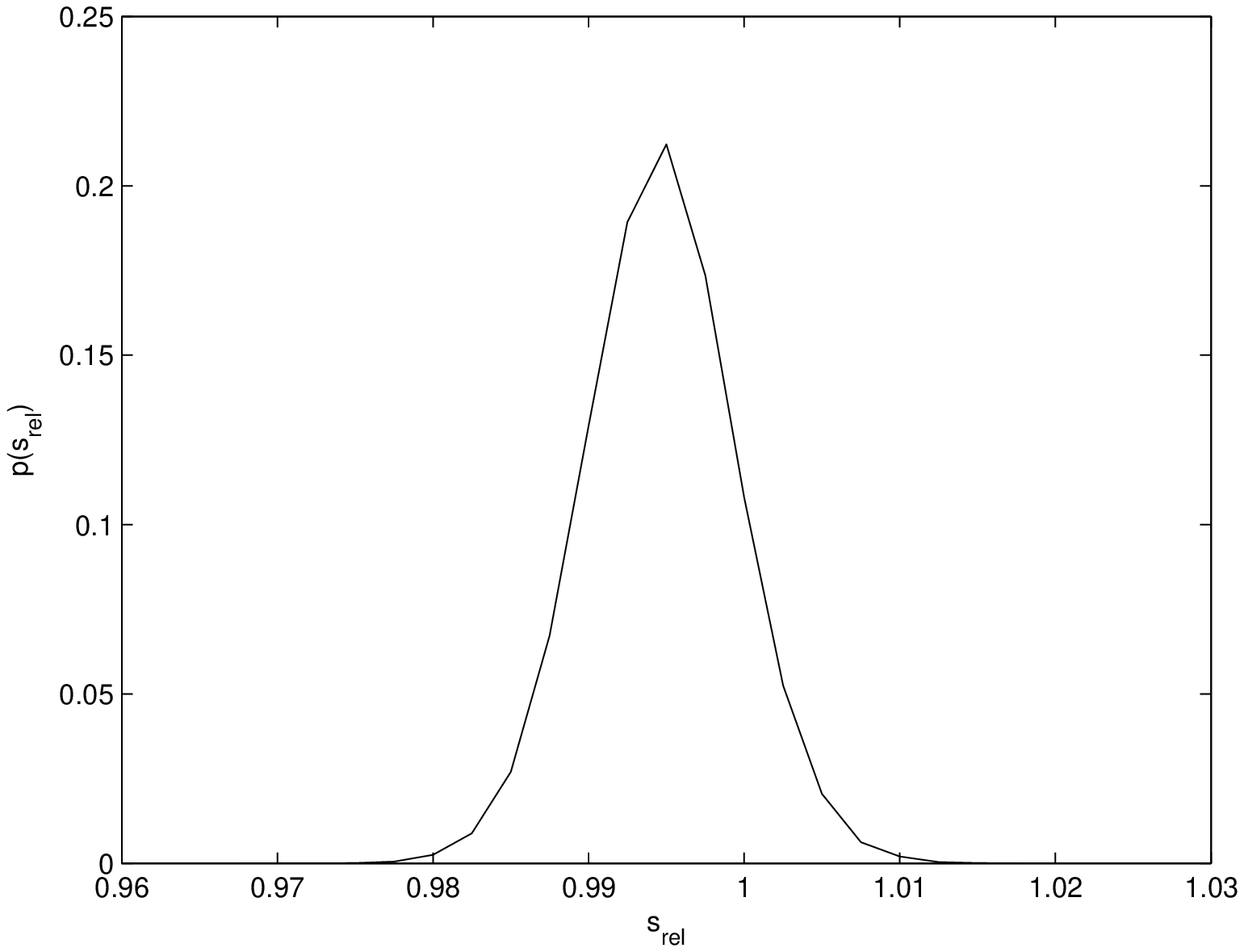}
\caption{The probability density function for $s_{\rm rel}$ for a
typical XCS catalogue with $a_i \in U(-5,5)$.} \label{fig:srel}
\end{center}
\end{figure}
Furthermore, finding typical $a_i$'s for the various XCS models, we
find that $s_{\rm rel} = 1 + \mathcal{O}(10^{-2})$.

In conclusion, the likelihood $\left\langle\ln{\cal
L}\right\rangle$ of the average catalogue is a good approximation to
the average likelihood $\ln \left\langle{\cal
L}\right\rangle_{\varepsilon}$ of our set of catalogues
$\mathcal{C}_{\pm}$, and can also generally be expected to be a good
approximation in other similar applications. We have confirmed
this by explicitly comparing to the likelihoods for a Poisson sample
of catalogues, as shown in Fig.~\ref{fig:contoursrand} in the main
text.

\bsp
\end{document}